\documentclass[journal,twoside,web]{ieeecolor}
\usepackage{generic}
\usepackage{cite}
\usepackage{amsmath,amssymb,amsfonts}
\usepackage{algorithmic}
\usepackage{graphicx}
\usepackage{textcomp}
\usepackage{epstopdf}
\usepackage{mathrsfs}
\usepackage{pifont}

\newcommand{\ts}{&  \hspace{-0.05in}}

\newcommand{\nn}{\nonumber}
\newcommand{\beq}{\begin{equation}}
\newcommand{\eeq}{\end{equation}}
\newcommand{\bea}{\begin{eqnarray}}
\newcommand{\eea}{\end{eqnarray}}
\newtheorem{thm}{\bf Theorem}
\newtheorem{lem}{\bf Lemma}
\newtheorem{defn}{\bf Definition}
\newtheorem{rem}{\bf Remark}

\usepackage[ruled,vlined]{algorithm2e}

\usepackage{enumerate}

\usepackage{float}

\usepackage{tikz}
\usetikzlibrary{positioning,matrix,arrows.meta,decorations.pathmorphing,fit,positioning,calc}
\usepackage{verbatim}
\usepackage{mathtools}
\usepackage{caption}
\usepackage{subcaption}
\usetikzlibrary{arrows}
\usepackage{lipsum}
\usepackage{tabularx}

\tikzset{
  state/.style={draw,circle,minimum size=14pt,inner sep=1pt},
  hyper/.style={draw,ellipse,inner sep=4pt},
  arr/.style={-Latex,thick},
  lab/.style={font=\scriptsize,fill=white,inner sep=1pt},
  layer/.style={font=\small}
}

\def\BibTeX{{\rm B\kern-.05em{\sc i\kern-.025em b}\kern-.08em
    T\kern-.1667em\lower.7ex\hbox{E}\kern-.125emX}}
\begin{document}
\title{Global and local observability of hypergraphs}
\author{Chencheng Zhang, 
        Hao Yang, \IEEEmembership{Senior Member, IEEE,}
        Shaoxuan Cui, 
        Bin Jiang, \IEEEmembership{Fellow, IEEE,}\\
        and~Ming Cao, \IEEEmembership{Fellow, IEEE}
\thanks{This work was supported in part by the National Natural Science Foundation of China under Grants 62233009 and 62403241, in part by Postdoctoral Fellowship Program of CPSF under Grant GZB20240975 and Jiangsu Provincial Outstanding Postdoctoral Program, in part by the National Key Laboratory Foundation of Helicopter Aeromechanics under Grant 2023-HA-LB-067-04, and in part by the Netherlands Organization for Scientific Research under Grant NWO-Vici-19902. }
\thanks{C. Zhang, H. Yang and B. Jiang are with the College of Automation Engineering, Nanjing University of Aeronautics and Astronautics, Nanjing, 211106, China (e-mail: zhangchencheng@nuaa.edu.cn; haoyang@nuaa.edu.cn; binjiang@nuaa.edu.cn). 
(Corresponding author: Hao Yang)}
\thanks{S. Cui is with the Bernoulli Institute for Mathematics, Computer Science and Artificial Intelligence, University of Groningen, Groningen, 9747 AG Netherlands  (E-mail:  s.cui@rug.nl)}
\thanks{M. Cao is with the Institute of Engineering and Technology (ENTEG), University of Groningen, Groningen, 9747 AG Netherlands  (e-mail: m.cao@rug.nl).
}}

\maketitle

\begin{abstract}
\textcolor{blue}{
This paper develops a unified algebraic-structural framework for analyzing global and local observability of dynamical network systems defined on non-uniform, directed, and weighted hypergraphs.
By exploiting polynomial rings and Hilbert's basis theorem, a tensor-based representation encodes group-wise interactions in dynamics and measurements, enabling Lie-derivative ideals with finitely many generators. Necessary and sufficient global observability conditions are established via ordered hyperedge contractions, linking each generator to a directed higher-order propagation path on the hypergraph. Building on this, higher-order network structure and weight design guidance are considered to achieve or reconstruct an observable system. Then, the structural aspects of observability are investigated by analyzing hyperedge-level reachability and dynamic-output hypergraph automorphisms. Finally, local observability conditions are developed whose rank tests explicitly reflect the 
higher-order hypergraph structure: different input/output configurations activate distinct 
tensor contractions and propagation paths. The effectiveness of the proposed criteria is demonstrated through a competitive population model with third-order interactions.}
\end{abstract}


\section{Introduction}
\label{sec:introduction}
\IEEEPARstart{N}{etwork} science has emerged and evolved over the past two decades, unveiling the underlying mechanisms of networked systems with interdisciplinary applications in systems engineering, systems biology, genetics, and neuroscience.\cite{Latora17}. Most research subjects focus on networks inherently depicted by only interactions between pairs of nodes \cite{Petri13}, whereas, in reality, the interactions within many systems go beyond dyadic relationships. It has been shown that functional brain networks \cite{Lee12}, protein interaction networks \cite{Estrada18}, and co-authorship graphs of publications \cite{Patania17} cannot be simply factorized by pairwise interactions.
Instead, they involve group-wise, or so-called higher-order, interactions among the nodes in the network. 
To bridge this gap, hypergraphs, where hyperedges link multiple nodes, are a potentially more powerful modeling tool to represent such kinds of interactions \cite{Carletti20}. 
For example, hypergraphs can be utilized to model transportation systems where routes involve multiple locations or intersections, allowing a more realistic depiction of traffic flows \cite{Huang23}.
Therefore, investigating the performance of hypergraphs has emerged at the forefront of research in network dynamics.

As one of the most important properties of a dynamical system, observability is a measure that tells whether the trajectory temporal evolution of the internal states of the system can be reconstructed from the knowledge of inputs and outputs \cite{kalman1959}. 
It can further indicate how to place sensors, as few as possible, to determine states \cite{Luenberger1966}. 
As for network systems, it can help position sensors on nodes to observe the states of all nodes.
Global and local observability of an initial state is defined according to whether the initial state is distinguishable from all the possible initial states or initial states in one of its neighborhoods. Most existing theoretical analysis for general nonlinear systems is based on algebra and differential geometry. Usually, the resulting approaches propose sufficient conditions for observability by computing the dimension of the subspace spanned by the gradients of the Lie derivatives of the measurements \cite{Hermann1977}\cite{Kawano2017}, which leads to the infinite number of Lie derivatives, especially for the global observability criterion for general nonlinear systems \cite{Zabczyk20}. 
However, for polynomial systems, it has been proved that global observability can be characterized by a finite set of equations based on commutative algebra \cite{Kawano13}. This can be further extended to analyze a class of nonlinear systems that can be transformed into polynomial expressions. 
Distinct from global observability, local case focuses on the observability properties in the vicinity of a specific state rather than across the entire state space of it \cite{Bar}.

\textcolor{blue}{
For network systems, the structure or interaction form plays an important role in affecting the properties of the systems, so as to observability.
Structural observability, originally developed for linear and pairwise networks \cite{Bellman1970, DiStefano1977}, has
been extended to networked and multi-agent systems to capture state identifiability properties that
hold for almost all numerical weights \cite{Angulo2019, Zhang2024}. These developments show that
propagation patterns and graph automorphisms play an important role in observability.}

\subsection{Problem Description}
\textcolor{blue}{
Due to the central role of hypergraphs in modeling higher-order interactions, a key open problem
is to determine how such group-wise couplings shape the system model through which information
propagates to the outputs. The challenge is not merely to decide whether a system is observable
in the classical sense, but to characterize, \emph{in finite and verifiable algebraic terms, how
directed higher-order hyperedges contribute to or impede state distinguishability}.
When moving beyond pairwise networks, \emph{structural observability}, which notion is
observability for almost all numerical hyperedge weights consistent with a given hypergraph
topology, becomes equally essential for sensor placement and model design. This raises a
fundamental question: \emph{which structural features of a hypergraph, its propagation patterns,
multi-way coupling orders, or inherent symmetries—govern the ability of outputs to distinguish
states?}  
Motivated by these considerations, this paper aims to establish algebraic and structural criteria
that characterize both global and local observability of higher-order network systems modeled
by hypergraphs, and to provide design principles for ensuring observability through output
selection and structural modification.}





\subsection{Related Work and Main Contribution}

\textcolor{blue}{
To the best of our knowledge, observability analysis for higher-order or hypergraph-based dynamical network systems remains largely unexplored. Existing
results are limited in scope: Ref.~\cite{Joshua23} considers only local observability for
uniform hypergraphs, and \cite{Chen24} provides tensor-rank tests for weak controllability and
observability of temporal hypergraphs. These approaches do not address global observability,
do not yield finite-step verifiable conditions, and do not cover general non-uniform, directed,
and weighted hypergraphs that arise in higher-order network dynamics. Moreover, none of the
existing work explains how higher-order propagation paths, tensor contractions, or hypergraph
symmetries fundamentally shape distinguishability. We try to \textit{establish the first comprehensive algebraic–structural framework for
global and local observability of higher-order network systems on general hypergraphs}. The main contributions are as follows:}

\begin{itemize}

    \item 
    \textcolor{blue}{
Inspired by the coupled cell system model on hypergraphs, a tensor algebra-based model is established for the higher-order network systems with inputs and outputs.
Both the dynamics and outputs of this system are on non-uniform, directed, and weighted hypergraphs, whose hyperedges involve different cardinalities. 
Thus, the established model is capable of capturing group-wise information in a network with groups of different sizes.}

\item
\textcolor{blue}{
By leveraging polynomial rings and Hilbert’s basis theorem, we derive necessary and sufficient 
global observability conditions that can be verified in finitely many steps. We further establish 
theorems that express the observability ideal through ordered tensor contractions along directed 
hyperedges, showing that each Lie–derivative generator corresponds to a higher-order propagation 
path. This reveals the relation between algebraic observability and the hypergraph structure.}


\item
\textcolor{blue}{
We connect the proposed algebraic observability criteria with the \emph{designability} of hypergraph structures and weights.  
(i) A specific situation that our criteria enable
fast certification of global observability is figured out:
in classical observability matrix rank-deficient cases, the Lie–derivative ideal chain stabilizes at very small order.
(ii) A guidance for reconstruction of unobservable network systems is proposed by adjusting higher-order couplings.
(iii) An efficient output design algorithm is established to guarantee the observability based on incremental Lie-derivative vanishing design.}

\item 
\textcolor{blue}{
We extend the analysis from numerical tensor conditions to purely structural ones for higher-order networks. Novel structural observability criteria are proposed by introducing observational diameter and hypergraph asymmetries from the structural perspective. The result shows that structural global observability holds if:
(i) the observational diameter is finite, which means that the dynamic–output hypergraph is topologically reachable, and  
(ii) the coupled hypergraph admits no nontrivial automorphism, which breaks symmetry structures.
This provides a structural characterization for observability independent of specific hyperedge coupling weights. }

   



\item 
\textcolor{blue}{
In the form of classical rank condition, three local observability criteria are respectively proposed for the higher-order network systems with input, with direct transmission, and without neither of them.
They identify the governing rule for how the multi-order derivatives of the system output evolve with the structure of higher-order networks.
This establishes the normalized relation between the observable matrices and the general hypergraph structures. }

\end{itemize}

The rest of the paper is organized as follows: Section II gives preliminaries of tensors and definitions of both global and local observability. Section III constructs general higher-order network dynamical systems on hypergraphs and transforms them into Kronecker product forms. The global observability criteria for the higher-order network systems are proposed in Section IV, which is followed by the local observability criteria presented in Section V. Section IV gives conclusions.

\section{Preliminaries}

\noindent\textbf{Notations : } 
Let $\mathbf{R}$ denote the field of real numbers, and $\mathbf{N}$ denote the set of non-negative integers. 
$\otimes$ represents the Kronecker product operator between matrices. 
${x}^{[k]}$ denotes the $k$-times Kronecker product of the vector $x$. 
A polynomial ring over $\mathbf{R}$ in the variables $x_i$ $(i = 1, 2, \dots, n)$ is denoted by $\mathbf{R}[x] := \mathbf{R}[x_1, \dots, x_n]$. 
Similarly, the polynomial rings over $\mathbf{R}$ in the variables $\xi_i$ $(i = 1, 2, \dots, n)$ and in both $\xi_i$ and $\eta_i$ $(i = 1, 2, \dots, n)$ are denoted by $\mathbf{R}[\xi] := \mathbf{R}[\xi_1, \dots, \xi_n]$ and $\mathbf{R}[\xi, \eta] := \mathbf{R}[\xi_1, \dots, \xi_n, \eta_1, \dots, \eta_n]$, respectively. 
The ideal generated by elements $a_1, a_2, \dots, a_s \in R[\xi, \eta]$ is written as $\langle a_1, \dots, a_s \rangle \subset \mathbf{R}[\xi, \eta]$. The sets $\mathbf{V}(I) \subset \mathbf{R}^n$ and $\mathbf{V}(J) \subset \mathbf{R}^n \times \mathbf{R}^n$ represent affine algebraic varieties, which are the sets of common zeros of all elements of the ideals $I \subset \mathbf{R}[\xi]$ and $J \subset \mathbf{R}[\xi, \eta]$, respectively.
For technical terms in commutative algebra and algebraic geometry, please refer to \cite{Atiyah, Cox1, Kunz, Cox2}.

\subsection{Tensors}
A tensor is a multidimensional array. The order of a
tensor is the number of its dimensions, and each dimension
is called a mode. A k-th order tensor usually is denoted by $\mathbf{T}\in \mathbf{R}^{n_1\times n_2\times \cdot\cdot\cdot\times n_k}$ and it is called cubical if all modes have same size, i.e., $n_1=n_2= \cdot\cdot\cdot=n_k=n$. A cubical tensor T is called super-symmetric if $\mathbf{T}_{j_1j_2\cdot\cdot\cdot j_k}$ is invariant under any permutation of the indices. To simplify the notation, we use symmetric to represent super-symmetric throughout the paper.

\begin{defn}[\!{\cite{Kolda09}}] \label{defn1} The tensor vector multiplication $\mathbf{T}\times_{p}\mathbf{v}$ along mode $p$ for a vector $\mathbf{v}\in \mathbf{R}^{n_k}$ is defined as 
\begin{align}
(\mathbf{T}\times_{p}\mathbf{v})_{j_1j_2\cdot\cdot\cdot j_{p-1}j_{p+1}\cdot\cdot\cdot j_k}=\sum_{j_p=1}^{n_p}\mathbf{T}_{j_1j_2\cdot\cdot\cdot j_p\cdots j_k}\mathbf{v}_{j_p}, 
\end{align}
which can be extended to 
\begin{align}
\mathbf{T}\times_{1}\mathbf{v}_1\times_{2}\mathbf{v}_2\cdot\cdot\cdot\times_{k}\mathbf{v}_k=\mathbf{T}\mathbf{v}_1\mathbf{v}_2\cdot\cdot\cdot\mathbf{v}_k \in \mathbf{R} \label{a11}
\end{align}
where $\mathbf{v}_p\in\mathbf{R}^{n_p}$ with $p\in\{1,...,k\}$. For simplicity, we define $\mathbf{T}\mathbf{v}_1\mathbf{v}_2\cdot\cdot\cdot\mathbf{v}_k :=\mathbf{T}\mathbf{v}^k$ if $\mathbf{v}_1=\cdots=\mathbf{v}_k=\mathbf{v}$.
\end{defn}

\subsection{Hypergraphs}

Some definitions regarding directed and undirected hypergraphs are introduced.  For more details, please see \cite{chen2021controllability} for undirected uniform hypergraphs and see \cite{gallo1993directed, xie2016spectral} for directed uniform hypergraphs. 

A weighted directed hypergraph is a triplet $\mathbf{H}=(\mathcal{V},\mathcal{E}, \Tilde{A})$. The set $\mathcal{V}$ denotes a set of vertices and $\mathcal{E}=\{E_1, E_2, \cdots,E_n \}$ is the set of hyperedges. A hyperedge is an ordered pair $E=(\mathcal{X},\mathcal{Y})$ of disjoint subsets of vertices; $\mathcal{Y}$ is the tail of $E$ and $\mathcal{X}$ is the head.
As a special case, a weighted undirected hypergraph is a triplet $\mathbf{H}=(\mathcal{V},\mathcal{E}, \Tilde{A})$, where $\mathcal{E}$ is a finite collection of non-empty subsets of $\mathcal{V}$. The cardinality of a hyperedge represents the number of nodes contained in it. If the cardinality of every hyperedge in a hypergraph is equal, then the hypergraph is uniform. Otherwise, the hypergraph is non-uniform. An illustration of uniform and non-uniform undirected hypergraphs is presented in Fig. 1.

\begin{figure}[h!]
        \centering
        \includegraphics[height=4.2cm]{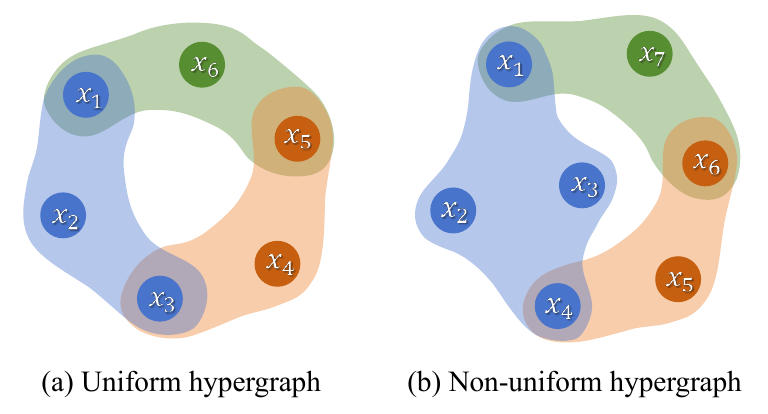}
        \caption{Illustration of uniform and non-uniform hypergraphs. 
        In (a), $e_1=\{x_1,x_2,x_3\}$, $e_2=\{x_3,x_4,x_5\}$, and $e_3=\{x_5,x_6,x_7\}$ are hyperedges. The cardinality of every hyperedge is equal to 3. 
        In (b), $\bar e_1=\{x_1,x_2,x_3,x_4\}$, $\bar e_2=\{x_4,x_5,x_6\}$, and $\bar e_3=\{x_1,x_6,x_7\}$ are hyperedges. The cardinality of hyperedge $\bar e_1$ is 4 and that of $\bar e_2$ and $\bar e_3$ is 3. The non-uniform hypergraph can be used to represent a network system containing different (higher-)order interactions. Please see Section III-B for details.
        }
        \label{fig:ill}
\end{figure}

\subsection{Observability of nonlinear systems}
Consider a nonlinear system in the following form
\begin{align}
\left\{\begin{aligned}   
\dot x&=F(x,u),x(0)=x_0  \\
y&=H(x,u) \label{a1}
\end{aligned}
\right.
\end{align}
where $x\in \mathbf{R}^n$, $u\in \mathbf{R}^m$, and $y\in\mathbf{R}^q$ denote state, input, output vectors. $F: \mathbf{R}^n\rightarrow\mathbf{R}^n$ and $H: \mathbf{R}^n\rightarrow\mathbf{R}^q$ are real analytic functions.

Definitions of observability are given here. For more details, please see references \cite{Kawano13}, \cite{Moog, Hermann77, Isidori95}.


\begin{defn} \label{defn2} A pair of initial states $(\xi,\eta)\in \mathbf{R}^n\times \mathbf{R}^n$ of the system (\ref{a1}) is distinguishable if there exists an admissible piecewise constant input $u$ and an instant $t$ such that $y(t;\xi,u)\neq y(t;\eta,u)$.
\end{defn}

\begin{defn} \label{defn3}  The system (\ref{a1}) is globally observable at an initial state $\xi\in \mathbf{R}^n$ if for any $\eta \in \mathbf{R}^n\setminus \{\xi\}$, the pair of $(\xi,\eta)$ is distinguishable.
\end{defn}

\begin{defn} \label{defn4} The system (\ref{a1}) is locally observable at an initial state $\xi \in M \subseteq \mathbf{R}^n$ if there exists an open neighborhood $U \subseteq M$ of $\xi$ such that for any open neighborhood $V \subseteq M$ of $\xi$ and $\eta \in V/ \{\xi\}$, the pair of $(\xi,\eta)$ is distinguishable.
\end{defn}

Note that here the definition of observability is depending on inputs. Definition \ref{defn3} represents that each pair of states can be distinguished by some input sequence that may depend on the initial states.
In this paper, the observability and local weak observability are called as global observability and local observability to distinguish the global and local cases.


\subsection{Problem formulation}


In the following sections III-VII, a general higher-order network control system is constructed on hypergraphs that can intuitively capture group-wise interactions, and the global and local observability testing methods are investigated for the system. To better demonstrate the logical relationship and structure of our main results, an outline of the proposed theorems is presented in Fig. \ref{fig3}.

\begin{figure}[h]
        \centering
        \includegraphics[height=6.7cm]{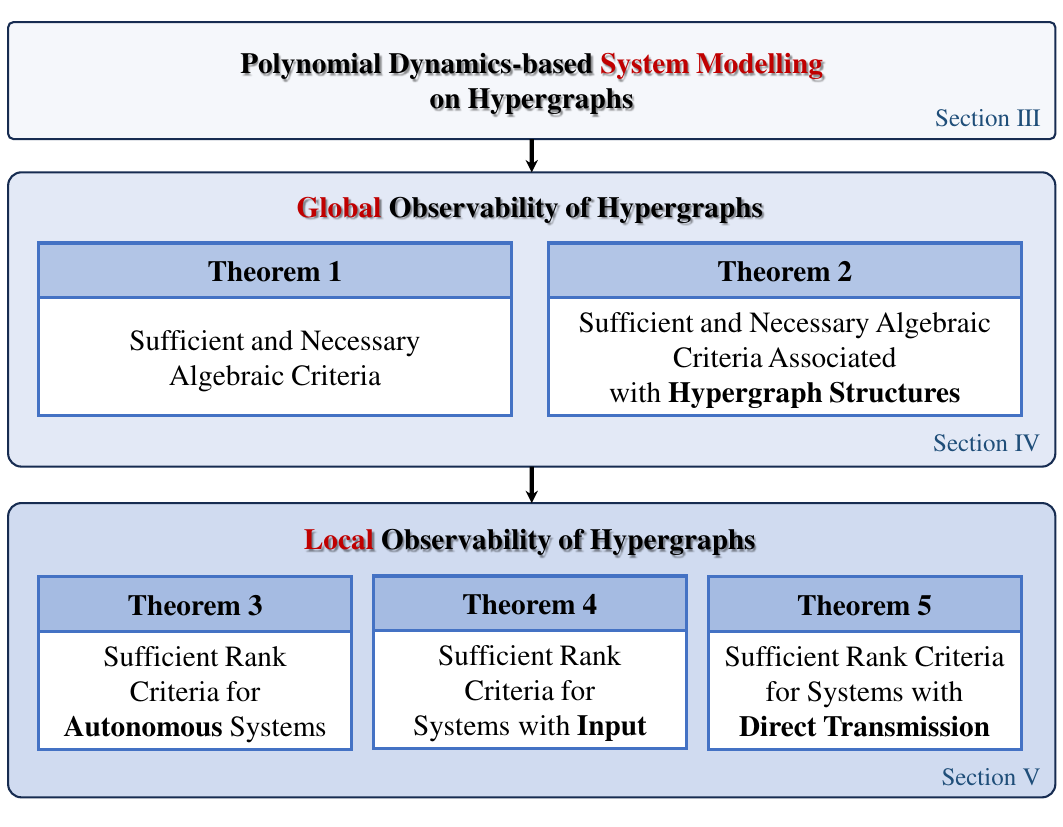}
        \caption{Outline of the proposed theorems.}
        \label{fig3}
\end{figure}

\section{Polynomial dynamics-based higher-order network systems on hypergraphs}

\subsection{Polynomial systems and relation to hypergraphs}

In this section, we illustrate the relationship between hypergraphs and polynomial dynamical systems through examples.

First of all, we take a coupled cell system \cite{bick2023higher, cui2024metzler} on a hypergraph as an example. Such systems can often be written in the form
\begin{equation}\label{eq:cell}
\begin{split}
    \dot{x}_i &= F(x_i) + \sum_{j=1}^N ( A_2)_{ji} \mathbf G_i(x_j, x_i) \\
    &\quad + \sum_{j, l=1}^N (\mathbf A_3)_{jli} \mathbf G_i^{(3)}(x_j, x_l, x_i) + \cdots
\end{split}
\end{equation}
where $F$ describes the intrinsic dynamics of the node, $A_2$ represents pairwise interactions between nodes, and higher-order interactions are captured by adjacency tensors $\mathbf A_s$ (with $s \geq 3$). For instance, $(\mathbf A_3)_{jli}$ and $\mathbf G_i^{(3)}( x_j, x_l, x_i)$ describe the joint influence of nodes $j$ and $l$ on node $i$. These interactions can be described by directed hyperedges with a single tail and multiple heads. 

To further illustrate, consider a system in which the coupling functions $\mathbf G_i^{(s)}( x_j, x_l, \dots, x_i)$ take a multiplicative form, such as $x_j x_l \dots x_i$. This form of interaction is commonly seen in physical systems, for example, where independent events occur simultaneously with probabilities given by multiplicative terms. The dynamical equation \eqref{eq:cell} can then be simplified and written in tensor form as:
\begin{align}
\dot{x} = \mathbf A_k x^{k-1} + \mathbf A_{k-1} x^{k-2} + \cdots + \mathbf A_2 x \label{ps1}
\end{align}
which mirrors the structure of many polynomial dynamical systems. It should be noted that any polynomial system, where the exponents are positive integers, can be written in this form. On the other hand, it is also worthwhile mentioning that if the interaction functions are not in a polynomial form but smooth, by the Taylor expansion, the system can be locally approximated by the polynomial system \eqref{ps1}. This further motivates us to study the polynomial system \eqref{ps1} on hypergraphs.

As is shown in \eqref{ps1}, hypergraphs provide a natural framework for representing polynomial dynamical systems, particularly those involving higher-order interactions. To demonstrate this, let us examine a few examples of real-world systems that utilize polynomial coupling functions and are modeled using hypergraphs. One such example is an epidemic model \cite{Epidemic, liang2024discrete, cui2023general}, where the spread of a disease through a population can be represented by a hypergraph. Each hyperedge captures how groups of individuals (nodes) interact and contribute to the transmission of the disease. This type of system can often be modeled using polynomial terms, such as the generalized Lotka-Volterra model, where the interaction between species in an ecological network can be described by polynomial coupling functions \cite{letten2019mechanistic, cui2023species}.  Inspired by the above discussion, a class of general hypergraph dynamical system model are proposed in the following section.

\subsection{General hypergraph dynamics modeling}
Consider a hypergraph $\mathcal{G}= (\mathcal{V},\mathcal{E})$ with $n$ nodes and $h$ hyperedges. Let $c$ be the maximum cardinality of the hyperedges. Define a general input-affine higher-order network control system on hypergraph $\mathcal{G}$, whose dynamics are as follows
\begin{align}\label{aa2}
\left\{\begin{aligned}      
\dot{x}=&\sum_{k=2}^{c}\mathbf{A}_k{x}^{k-1}+\sum_{j=1}^{m}\sum_{k=2}^{c}\mathbf{B}_{k,j}{x}^{k-1}u_j  \\
y_i=&\sum_{k=1}^{c}\mathbf{C}_{i,k}{x}^{k}+\sum_{l=1}^{w_i}\sum_{k=1}^{c}\mathbf{D}_{i,k,l}{x}^{k}u_l 
\end{aligned}   \right.
\end{align}
where $i=1,\cdots,q$, ${x}:=(x_1,x_2,...,x_n)^T$, and ${x}(0)={x}_o$. $u_j$ and $y_i$ are the $j$-th input and the $i$-th output where $j=1,\cdots,m$ and $i=1,\cdots,q$. 
$\mathbf{A}_k,\mathbf{B}_k,\textbf{C}_{i,k},\mathbf{D}_{i,k,l} \in \mathbf{R}^{n\times n\times\cdots\times n}$ denote $k$-th order system, input, output and direct transmission tensors, respectively. In different hypergraphs, these tensors are all adjacency tensors that capture the strength of interactions among nodes.
Define $e_{\mathbf{A}kr}$, $e_{\mathbf{B}kjr}$, $e_{\mathbf{C}ikr}$ and $e_{\mathbf{D}iklr}$ as the $r$-th hyperedge with $k$ cardinality where $r=1,2,\cdots$ in coupled dynamics of subsystems corresponding to the system ${x}$, input $u_j$, output $y_i$ and direct transmission $u_l$ for output $y_i$, respectively. Since we focus on nodal dynamics, it is sufficient to consider hypergraphs where all hyperedges have a single tail. The details of the setting will be explained later in this section.  
$\mathbf{A}_k, \mathbf{B}_{k,j}, \mathbf{C}_{i,k}$ and $\mathbf{D}_{i,k,l}$ are defined as follows
\begin{align}\nonumber
{(\mathbf{A}_k)}_{a_1\cdots a_k}:=\left\{\begin{array}{cc}
&\hspace{-6mm} \frac{{({\tilde{\mathbf{A}}}_k)}_{a_1\cdots a_k}}{(k-1)!},
\left(\{a_1,\cdots,\! a_{k-1}\}, \{a_k\}\right) \in \!E_{\mathbf{A}k}\\
&\hspace{-3cm} 0, \ \ \mathbf{otherwise}
\end{array}
\right.
\end{align}

\begin{align}\nn
{(\mathbf{B}_{k,j})}_{b_1\cdots b_k}:=\left\{\begin{array}{cc}
&\hspace{-6mm} \frac{{({\tilde{\mathbf{B}}}_{k,j})}_{b_1\cdots b_k}}{(k-1)!},
\left(\{b_1,\cdots,\! b_{k-1}\}, \{ b_k\}\right) \in \!E_{\mathbf{B}kj}
\\
&\hspace{-3cm} 0, \ \ \mathbf{otherwise}
\end{array}
\right.
\end{align}

\begin{align}\nn
{(\mathbf{C}_{i,k})}_{c_1\cdots c_k}:=\left\{\begin{array}{cc}
&\hspace{-6mm}\frac{{({\tilde{\mathbf{C}}}_{i,k})}_{c_1\cdots c_k}}{k!}, 
\left(\{c_1,\cdots,\! c_{k-1}\}, \{ c_k\}\right) \in \!E_{\mathbf{C}ik}
\\
&\hspace{-3cm} 0,   \ \mathbf{otherwise}
\end{array}
\right.
\end{align}

\begin{align}\nn
{(\mathbf{D}_{i,k,l})}_{d_1\cdots d_k}\!\!:=\!\left\{\begin{array}{cc}
&\hspace{-6mm}\frac{{({\tilde{\mathbf{D}}}_{i,k,l})}_{d_1\cdots d_k}}{k!},  \left(\{d_1,\cdots,\! d_{k-1}\}, \{d_k\}\right) \in \!E_{\mathbf{D}ikl}
\\
&\hspace{-3.2cm} 0,   \ \mathbf{otherwise}.
\end{array}
\right.
\end{align}
where 
$E_{\mathbf{A}k}:=\{e_{\mathbf{A}kr}|r=1,2,\cdots \}$,
$E_{\mathbf{B}kj}:=\{e_{\mathbf{B}kjr}|r=1,2,\cdots \}$,
$E_{\mathbf{C}ik}:=\{e_{\mathbf{C}ikr}|r=1,2,\cdots \}$, and $E_{\mathbf{D}ikl}:=\{e_{\mathbf{D}iklr}|r=1,2,\cdots\}$.

From Definition \ref{defn1}, one can further obtain that 
\begin{equation}
(\mathbf{A}_kx^{k-1})_i=\sum_{i_1,\cdots,i_{k-1}=1}^{n}{(\mathbf{A}_k)}_{i_1\cdots i_{k-1}i}x_{i_1}\cdots x_{i_{k-1}} \label{product}
\end{equation}
According to (\ref{product}), it is evident that $\mathbf{A}_kx^{k-1}$ can represent any arbitrary homogeneous polynomial vector field.

In the above, we design a general higher-order network system framework on hypergraph \eqref{aa2} based on tensors. Utilizing tensors to represent hyperedges allows for a compact and structured way to {\color{blue}capture the combined influence of all head nodes on tail nodes in a hypergraph}. Specifically, each entry in the tensor corresponds to a hyperedge, with indices representing the involved nodes. 
In the above weighted and directed setting, the tensor encodes how the group of head nodes collectively influences a tail node. This modeling approach effectively generalizes pairwise interactions (as seen in standard graphs) to higher-order interactions in hypergraphs, where multiple nodes simultaneously impact one node through a single hyperedge, i.e., $\left(\{a_1,\cdots,\! a_{k-1}\}, \{a_k\}\right) $ represents the hyperedge with a tail $a_k$ and heads $a_1,\cdots,a_{k-1}$.  ${(\mathbf{A}_k)}_{a_1\cdots a_k}$ represents the weight of the joint influence acted on $a_k$ by $a_1,\cdots,a_{k-1}$.

Additionally, when the hypergraph is undirected and the weight tensors ${({\tilde{\mathbf{A}}}_k)}_{a_1\cdots a_k}$, $ {({\tilde{\mathbf{B}}}_{k,j})}_{b_1\cdots b_k}$, $ {({\tilde{\mathbf{C}}}_{i,k})}_{c_1\cdots c_k}$, ${({\tilde{\mathbf{D}}}_{i,k,l})}_{d_1\cdots d_k}$ are all set to 1, the proposed system design framework guarantees that all hyperedges are unweighted, resulting in a simplified structure. Specifically, the removal of weights allows for a more direct representation of the hypergraph’s interactions without additional complexity. This method underscores the flexibility of tensor-based modeling for capturing both weighted and unweighted higher-order network interactions, effectively adapting to various scenarios in hypergraph representation. To illustrate the modeling process vividly, Fig. \ref{fig2} is presented.

\begin{figure}[h]
        \centering
        \includegraphics[height=3.3cm]{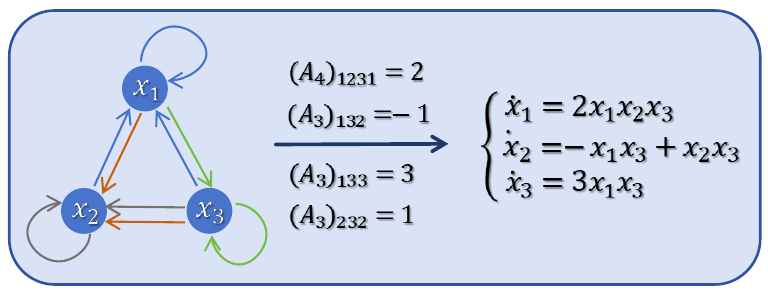}
        \caption{Illustration of higher-order network system modeling. 
        In the left-side hypergraph, arrays in the same color are used to capture a group of interactions. For example, the blue arrays represent the joint influence from $x_1$, $x_2$ and $x_3$ acted on $x_1$ such that there is a directed hyperedge $(\{x_1,x_2,x_3\},\{x_1\})$; the orange arrays represent the joint influence from $x_1$ and $x_3$ acted on $x_2$ such that there is a hyperedge $(\{x_1,x_3\},\{x_2\})$, and so on. 
        Furthermore, the weight of each hyperedge is designed as the corresponding tensor element. This represents the right-side higher-order network system.  
        }
        \label{fig2}
\end{figure}

\begin{rem} \label{rem1}
Compared with the previous research \cite{chen2021controllability, Chen24, Joshua23}, the  established higher-order network system (\ref{aa2}) broadens the scope of higher-order network structures that can be effectively represented. First of all, the system (\ref{aa2}) is modeled on non-uniform hypergraphs rather than uniform ones. Specifically, the system (\ref{aa2}) is modeled by dividing it into several uniform hypergraphs described by tensors with different orders. For every $k\in\{1,\cdots, c\}$, $\mathbf{A}_k{x}^{k-1}$ can be viewed as all the hyperedges with $k$ nodes in the hypergraph. Secondly, not only in system dynamics, but also that in system input and output units,  the interaction between nodes can be higher-order. They are captured with different hypergraphs. In details,
while the dynamics of any node is dominated by all the other nodes that be contained in hyperedges of dynamics hypergraphs, i.e., $\mathbf{A}_k{x}^{k-1}$, any output in the system is dominated by all the hyperedges in output hypergraphs, i.e., $\mathbf{C}_k{x}^{k}$. This setting is necessary
since the output of many real-word systems depends not just on individual node states but also on their collective or group behaviors. For example, 
in an economic system, the price of a product may depend not only on individual factors (e.g., supply and demand) but also on complex market-wide phenomena that involve multiple agents interacting simultaneously, which can be effectively modeled with hyperedges. 
Overall, system (\ref{aa2}) allows to accurately capture more dynamics, offering a comprehensive insight into the network system's overall functioning.
Once $k=1$, both the dynamics and output reduce to the classical pairwise interacted network case.
\end{rem}

\begin{rem} \label{rem2}According to the above analysis, one can figure out that higher-order network system \eqref{aa2} can be expressed in the polynomial form, this further leads us to investigate the observability problems of hypergraphs by virtue of properties of polynomials.  
Please note that although many polynomial systems can be mapped to hypergraphs, 
there exist structural constraints of hypergraphs such as the requirement of defining interactions through hyperedges. The relation between hypergraphs and polynomial systems is presented in Section III-A in detail. In particular, any homogeneous polynomial system can be expressed as $\dot{x} = \mathbf A x^{k-1}$, which corresponds to a uniform hypergraph. \end{rem}

\subsection{Kronecker form}
Note that system (\ref{aa2}) can be represented in its Kronecker form as 
\begin{eqnarray}\label{a3}
             \left\{\begin{aligned} 
\dot{{x}}=&\sum_{k=2}^{c}\mathcal{A}_k{x}^{[k-1]}+\sum_{j=1}^{m}\sum_{k=2}^{c}\mathcal{B}_{k,j}{x}^{[k-1]}u_j \\
y_i=&\sum_{k=1}^{c}\mathcal{C}_{i,k}{x}^{[k]}+\sum_{l=1}^{w_i}\sum_{k=1}^{c}\mathcal{D}_{i,k,l}{x}^{[k]}u_l \end{aligned}\right.
\end{eqnarray}
where $i=1,\cdots,q$ and ${x}(0)={x}_o$.
$\mathcal{A}_k=\mathbf{A}_k\in\mathbf{R}^{n\times n^{k-1}}$,  $\mathcal{B}_{k,j}=\mathbf{B}_{k,j}\in\mathbf{R}^{n\times n^{k-1}}$, $\mathcal{C}_{i,k}=\mathbf{C}_{i,k}\in\mathbf{R}^{1\times n^{k}}$, and $\mathcal{D}_{i,k,l}=\mathbf{D}_{i,k,l}\in\mathbf{R}^{1\times n^{k}}$ are unfolding of $\mathbf{A}$, $\mathbf{B}$, $\mathbf{C}$ and $\mathbf{D}$.
Moreover, for simplicity of notation, the system (\ref{a3}) can be denoted as the following integrated form 
\begin{align}\label{intsys}
            \left\{\begin{aligned} 
\dot{{x}}=&f({x}, u), \ \ {x}(0)={x}_o \\ 
y_i=&h_i({x}, u), \ \ i=1,2,\cdots,q  \end{aligned}\right.
\end{align}
where 
\begin{align}\label{aa4}
\left\{\begin{aligned} 
f({x}, u):=g_0({x})+\sum_{j=1}^{m}g_j({x})u_j \\ 
h_i:=p_{i,0}({x})+\sum_{l=1}^{w_i}p_{i,l}({x})u_l \end{aligned}\right.
\end{align}
with $g_0({x}):=\sum_{k=1}^{c}\mathcal{A}_k{x}^{[k-1]}$, $g_j({x}):=\sum_{k=1}^{c}\mathcal{B}_{k,j}{x}^{[k-1]}$ for $j=1,\cdots,m$
and $p_{i,0}({x}):=\sum_{k=1}^{c}\mathcal{C}_{i,k}{x}^{[k]}$ and $p_{i,l}({x}):=\sum_{k=1}^{c}\mathcal{D}_{i,k,l}{x}^{[k]}$ for $l=1,\cdots,w_i$.

\section{Global observability criteria of hypergraphs}

Recall the relationship between hypergraphs and polynomial systems as described in Section III. In this section, inspired by the theorems on observability for polynomial systems proposed by \cite{Kawano13}, we present an global observability criterion for hypergraphs. This criterion is formulated based on the symmetrization of tensors and the underlying structure of the hypergraph.

\subsection{Global observability of higher-order networks}

Based on Definition \ref{defn2}, the definition of indistinguishability for initial states of system \eqref{aa2} can be represented as follows. 

\begin{defn} \label{defn5}
A pair of initial states $(\xi,\eta)\in \mathbf{R}^n\times \mathbf{R}^n$ of system (\ref{aa2}) is indistinguishable if $y(t;\xi,u)= y(t;\eta,u)$ holds for any admissible piecewise constant input $u$ and instant $t$.
\end{defn}

According to the above definition, the indistinguishable initial states can be obtained by the following lemma.

\begin{lem} \label{lem1} For the higher-order network system \eqref{aa2}, a pair of initial states $(\xi,\eta)$ is indistinguishable if and only if, for all $\tau_1,\tau_2,\cdots,\tau_r\in\{g_0,g_1,\cdots g_m\}, r=0,1,\cdots$ and $\upsilon_c\in \{p_{i,0},p_{i,1},\cdots,p_{i,w_i}\},i=1,\cdots,q$, $c=0, \cdots,\sum_{i=1}^{q}(w_i+1)-1$, it holds that
\begin{align}
L_{\tau_1}L_{\tau_2}\cdots L_{\tau_r}\upsilon_c(\xi)=L_{\tau_1}L_{\tau_2}\cdots L_{\tau_r}\upsilon_c(\eta)      \label{cond1}
\end{align}
where the Lie derivative $L_\tau v:=\frac{\partial v}{\partial x}\tau$.
\end{lem}

\begin{proof}For system (\ref{aa2}), consider the following piecewise constant input
\begin{align}
u_k=\left\{
\begin{array}{ccc}
     u_{k1}& t\in[0,t_1)  \\
     u_{k2}& t\in[t_1,t_1+t_2)\\
     \vdots& 
\end{array} 
\right.
\end{align}
where  $t_1, t_2,\cdots$ are assumed to be sufficiently small. Since $h_1, h_2, \cdots, h_q$ are real analytic functions of $t$, the output $y_i$ at time $t \in \{t_1, t_1 + t_2, \cdots\}$ can be computed in the following Taylor series form
\begin{align} \label{Ts}
y_i(t_1)&=h_i(x_0,u_{k1})+L_{f(x,u_{k1})}h_{i}(x,u_{c1})|_{x_0}t_1+\cdots \nn\\
y_i(t_1+t_2)&=h_{i}(x_0,u_{k1})+L_{f(x,u_{k1})}h_j(x,u_{k1})|x_0t_1\\
&\hspace{5mm}+L_{f(x,u_{k2})}h_j(x,u_{k2})|_{x_0}t_2+\cdots \nn \\
&\hspace{0.3cm}\vdots \nn
\end{align}

\textit{Necessity:}
Since system (\ref{a3}) can be computed as a polynomial system, the functions $f, h$ are real analytic. Therefore, a pair of initial states $(\xi,\eta)$ is indistinguishable if for all $i=1,\cdots,q$, it holds that
\begin{align}\label{aa5}
\begin{array}{ccc}
    &h_{i}(x,u_{k1})|_{\xi}=h_{i}(x,u_{k1})|_{\eta}\\
&L_{f(x,u_{k1})}h_{i}(x,u_{c1})|_{\xi}=L_{f(x,u_{k1})}h_{i}(x,u_{c1})|_{\eta} \\
&\vdots 
\end{array}
\end{align}
Since for any possible $u$, the above equations (\ref{aa5}) should hold. It follows from (\ref{aa4}) that equations (\ref{aa5}) lead to    
\begin{align}\label{aa6}
\begin{array}{ccc}
&p_{i,k}(\xi)=p_{i,k}(\eta)  \\
&L_{g_r}p_{i,k}(\xi)=L_{g_r}p_{i,k}(\eta)  \\
\hspace{0.2cm} &\vdots
\end{array}
\end{align}
with $r=0,1,\cdots,m$ and $k=0,1,\cdots,w_i$, which can be incorporated as condition (\ref{cond1}).

\textit{Sufficiency:} Conversely, consider that condition (\ref{cond1}) is satisfied. By adding Lie derivatives according to the Taylor series \eqref{Ts}, one can obtain $y_i$, $i=1,\cdots, q$ for system \eqref{aa2}. Thus, \eqref{aa5} is satisfied. This further leads to $y_i(t;\xi,u)=y_i(t;\eta,u)$ for any piecewise $u$ and instant $t$. Therefore, the initial pair $(\xi, \eta)$ is indistinguishable according to Definition \ref{defn5}.
\end{proof}

\begin{rem} \label{rem3}According to Lemma \ref{lem1}, one can conclude that the set of initial states that satisfy condition (\ref{cond1}) are not observable.  However, it is difficult to determine how many Lie derivatives are sufficient to test the observability. To solve this problem, relying on Hilbert’s basis theorem, the following observability criteria can be obtained using a finite set of equations by using properties of polynomial rings \cite{Kawano13}.
\end{rem}

From condition (\ref{cond1}), if the initial state pair $(\xi,\eta)$ is indistinguishable, it holds that 
\begin{align}
v_c(\xi) &=v_c(\eta)  \nn \\
L_{g_j}v_c(\xi) &=L_{g_j}v_c(\eta) \nn \\
L_{g_k}L_{g_j}v_c(\xi) &=L_{g_k}L_{g_j}v_c(\eta)  \nn \\
&\hspace{2mm}\vdots  \nn
\end{align}
where $c=1, \cdots,\sum_{i=1}^{q}(w_i+1)$, $j\in 1,\cdots,m$, and $k\in 1,\cdots,m$. 
For $i=1,\cdots, q$ and $r=0, 1, \cdots$, define the ideals $J_{r,i}$ as
\begin{align}
J_{0,i}&=\langle v_{i_1}(\xi)-v_{i_i}(\eta), \cdots, v_{i_{s_i}}(\xi)-v_{i_{s_i}}(\eta)\rangle   \nn\\
J_{1,i}&=J_{0,i}+\langle L_{g_1}v_{i_1}(\xi)-L_{g_1}v_{i_1}(\eta), \nn\\
&\hspace{3cm}\cdots, L_{g_m}v_{i_{s_i}}(\xi)-L_{g_m}v_{i_{s_i}}(\eta)\rangle  \nn\\
J_{2,i}&=J_{1,i}+\langle L_{g_1}L_{g_1}v_{i_i}(\xi)-L_{g_1}L_{g_1}v_{i_1}(\eta), \nn\\
&\hspace{2cm}\cdots, L_{g_m}L_{g_m}v_{i_{s_i}}(\xi)-L_{g_m}L_{g_m}v_{i_{s_i}}(\eta)\rangle  \nn \\
&\hspace{2mm}\vdots  \nn \\
J_{r,i}&=J_{r-1,i}+\langle L_{g_1}\cdots L_{g_1}v_{i_i}(\xi)-L_{g_1}\cdots L_{g_1}v_{i_1}(\eta),\nn\\
&\hspace{1cm} \cdots, L_{g_m}\cdots L_{g_m}v_{i_{s_i}}(\xi)-L_{g_m}\cdots L_{g_m}v_{i_{s_i}}(\eta)\rangle  \nn 
\end{align}
where $i_1:=\sum_{c=1}^{i-1}w_c+i-1$ and $i_{s_i}:=\sum_{c=0}^{i}w_c+i-1$.

Further, define $J_r:=\sum_{i=1}^qJ_{r,i}$. 
It holds that $J_0\subset{J_1}\subset{J_2}\subset \cdots$. According to Hilbert's basis theorem and the ascending chain condition \cite{Cox1}, there exists a certain ideal satisfying 
\begin{align}
J_0\subset{J_1}\subset{J_2}\subset \cdots \subset J_N=J_{N+1}=\cdots =J \label{JN}
\end{align}
for a certain number $N$. Define  $\phi:\mathbf{R}[\xi,\eta]\rightarrow\mathbf{R}[\xi]$ as the mapping that substitutes a known initial state vector $\sigma$ into $\eta$. Denote the ideal $\mathcal{J}\subset \mathbf{R}[\xi]$ as the image $\phi(J)$, which characterizes the set of initial states that are indistinguishable from $\sigma$. $\mathcal{J}$ is generated by a finite set of polynomials and its algebraic variety $\mathbf{V}(\mathcal{J})$ consists of the common zeros of those finite polynomials.

To help test the observability of the initial state $\sigma:=(\sigma_1,\sigma_2,\cdots,\sigma_n)^T$ for the system (\ref{aa2}), define the following ideal $\ell\subset\mathbf{R}[\xi]$ as 
\begin{align}
    \ell :=\langle \xi_1-\sigma_1,\xi_2-\sigma_2,\cdots,\xi_n-\sigma_n\rangle 
\end{align}
The algebraic variety of $\ell $ is $\mathbf{V}(\ell)=\{\sigma\}$. 

Before going on, we give the following lemma of global observability for system \eqref{aa2} directly from Definition \ref{defn3}.

\begin{defn} \label{defn6}The system (\ref{aa2}) is globally observable at an initial state $\xi\in \mathbf{R}^n$ if for any $\eta \in \mathbf{R}^n\setminus \{\xi\}$, the pair of $(\xi,\eta)$ is distinguishable.
\end{defn}

Then the following Lemma can be obtained.

\begin{lem} \label{lem2} The higher-order network system (\ref{aa2}) is globally observable at an initial state $\sigma \in \mathbf{R}^n$ if and only if one of the following equivalent
conditions hold
\begin{align}
&(a) \ \mathbf{V}(\mathcal{J})=\mathbf{V}(\ell) \nn\\
&(b) \ \mathbf{V}(\sqrt{\mathcal{J}}:\ell)\subset\mathbf{V}(\ell) \nn
\end{align}
\end{lem}

\begin{proof}
Consider condition (a). We prove sufficiency and necessity separately as follows.

\textit{Sufficiency:} Since $\mathbf{V}(\ell)=\{\sigma\}$, condition (a) is equivalent to $\mathbf{V}(\mathcal{J})=\{\sigma\}$.  The variety $\mathbf{V}(\mathcal{J})$ is the set of initial states that are indistinguishable from $\sigma$, which only contains $\sigma$. This means that $\sigma$ is distinguished. Thus, condition (a) implies that $\sigma$ is observable according to Definition \ref{defn6}. 

\textit{Necessity:} If the initial state $\sigma$ is observable, which means it is distinguishable according to Definition \ref{defn6}, the variety $\mathbf{V}(\mathcal{J})$ should be $\{\sigma\}$. Since $\mathbf{V}(\ell)=\{\sigma\}$, it holds that $\mathbf{V}(\mathcal{J})=\mathbf{V}(\ell)=\{\sigma\}$. Condition (a)  holds. Therefore, condition (a) is a necessary and sufficient condition for global observability of $\sigma$ according to Definition \ref{defn6}.

Consider condition (b). We prove that condition (a) is equivalent to condition (b).

\textit{Condition (a) implies condition (b): }Since for any two ideals $J, I\subset \mathbf{R}[\sigma]$, the following properties hold \cite{Cox1}
\begin{align}
    \mathbf{V}(\sqrt{J}:I)&\subset \mathbf{V}(\sqrt{J})\label{p1}\\
    \mathbf{V}(\sqrt{J})&=\mathbf{V}({J})\label{p2}\\
    \mathbf{V}(\sqrt{J})\setminus  \mathbf{V}(I)&\subset \mathbf{V}(\sqrt{J}:I)\label{p3}
\end{align}
Combining properties (\ref{p1}) and (\ref{p2}), one can deduce that $\mathbf{V}(\sqrt{\mathcal{J}}:\ell)\subset\mathbf{V}(\sqrt{\mathcal{J}})=\mathbf{V}(\mathcal{J})$. This means that condition (a) implies condition (b). 

\textit{Condition (b) implies condition (a): } Inversely, according to properties (\ref{p2}) and (\ref{p3}), it holds that 
\begin{align}
 \mathbf{V}(\mathcal{J})\setminus \mathbf{V}(\ell) = \mathbf{V}(\sqrt{\mathcal{J}})\setminus \mathbf{V}(\ell) \subset \mathbf{V}(\sqrt{\mathcal{J}}:\ell) \nn
\end{align}
If condition (b) holds, it yields that 
\begin{align}
 \mathbf{V}(\mathcal{J})\setminus \mathbf{V}(\ell) \subset \mathbf{V}(\ell)  \nn
 \end{align}
 This leads to $\mathbf{V}(\mathcal{J})=\mathbf{V}(\ell)$. Therefore, condition (b) implies condition (a) as well. This completes the proof.
\end{proof}

Based on the Gröbner basis and Hilbert's Nullstellensatz theories \cite{Cox1}, the following sufficient conditions can be deduced by Lemma \ref{lem1}.

\noindent\textbf{Corollary 1 : }The higher-order network system (\ref{aa2}) is globally observable at an
initial state $\sigma \in \mathbf{R}^n$ if one of the following conditions holds 
\begin{align}
   & (a) \mathbf{V}(\sqrt{\mathcal{J}}:\ell)=\emptyset \nn\\
   & (b)\mathcal{J}=\ell\nn\\
   & (c)\sqrt{\mathcal{J}}:\ell=\mathbf{R}[\xi]\nn\\
   & (d)1\in \sqrt{\mathcal{J}}:\ell \nn
\end{align}
\hspace*{\fill}$\Box$

\begin{rem} \label{rem4}By transferring system on higher-order networks into polynomial forms, conditions (a)-(b) in Corollary 1 are necessary and sufficient conditions to test global observability of higher-order network system  (\ref{aa2}). The basic idea is to check the common zeros of a finite set of equations given by generators of ideal $\mathcal{J}$ or $(\sqrt{\mathcal{J}}:\ell)$. By virtue of the properties of the polynomial ring and Hilbert's basis theorem, the global observability criteria are testable in finite steps different from the condition in Lemma \ref{lem1}, which makes it difficult to determine how many Lie derivatives should be calculated. 
Conditions (a)-(d) in Corollary 2 are sufficient conditions for testing the global observability of a higher-order network system (\ref{aa2}), which is easier to verify compared with those in Corollary 1. Conditions (c)-(d) can be verified using the Gröbner basis \cite{Cox1}.
\end{rem}

\subsection{Global Observability of Higher-order Networks associated with Hypergraph Structures}

The algebraic conditions for detecting the global observability of system \eqref{aa2} can be traced back to the structure of the hypergraph, specifically in relation to the system's tensor representation. In Lemma \ref{lem2} and Corollary 1, the primary challenge lies in the computation of ideal
$\mathcal{J}$, which is typically difficult to obtain.

However, in the case of the hypergraph, this difficulty is alleviated by expressing 
$\mathcal{J}$ in a generalized form using tensors in system \eqref{aa2}.
These tensors describe the group-wise connection behaviors of a hypergraph and can reveal its structure, thereby associating the observability criterion with hypergraph structures. On the other hand, by presenting a general expression of the algebraic condition based on system tensors, the computation is also simplified.

Firstly, we consider the following higher-order network system under symmetric hypergraph structure
\begin{align}\label{simplified}
            \left\{\begin{aligned} 
\dot{x}&=\mathbf{A}{x}^{k-1}\\
y&=\mathbf{C}{x}^{k} 
\end{aligned}\right.
\end{align}
where $\mathbf{A}$ and $\mathbf{C}$ are $n$-dimensional $k$-th order tensors.


Before going on, the following two lemmas are given to help investigate the observability based on the hypergraph structure. 


\begin{lem} [\!\cite{chen2022explicit}]
Given a one dimensional homogeneous polynomial function $f(x(t))$: $\mathbb{R}^n\rightarrow\mathbb{R}.$ It can be uniquely determined by $\mathbf{A}x^m$, where $m$ is the order of tensor $\mathbf{A}$, $\mathbf{A}$ is symmetric. 
\hspace*{\fill}$\Box$

Let $\mathbf{A}\in \mathbb{R}^{[m, n]}$. Then there is a unique symmetric tensor $\mathbf{B} \in \mathbb{R}^{[m, n]}$ such that for all $\mathbf{A}x^m = \mathbf{B}x^m$.
We call $\mathbf{B}$ the symmetrization of $\mathbf{A}$, and denote it $\operatorname{sym}(\mathbf{A})$. 
\end{lem}

\begin{lem}[Symmetrization of tensors] \label{lem3}
Define a permutation of indexes $j_1, \cdots, j_m$ if $\left\{k_1, \cdots, k_m\right\}=\left\{j_1, \cdots, j_m\right\}$ as $\left(k_1, \cdots, k_m\right)$. Denote permutation operation by $\sigma$, i.e., $\sigma\left(j_1, \cdots, j_m\right)=$ $\left(k_1, \cdots, k_m\right)$. Denote the set of all distinct permutations of an index set $\left(j_1, \cdots, j_m\right)$ by $\Sigma\left(j_1, \cdots, j_m\right)$. Note that $\left|\Sigma\left(j_1, \cdots, j_m\right)\right|$, the cardinality of $\Sigma\left(j_1, \cdots, j_m\right)$, is variant for different index sets. For example, if $j_1=\cdots=j_m$, then $\left|\Sigma\left(j_1, \cdots, j_m\right)\right|=1$; but if all of $j_1, \cdots, j_m$ are distinct, $\left|\Sigma\left(j_1, \cdots, j_m\right)\right|=m!$.
Let ${\mathbf{A}} \in \mathbb{R}^{[m, n]}$. Then, the symmetrization of tensor $\mathbf{A}$ is
\begin{equation}\nn
   {\operatorname{sym}({\mathbf{A}})}_{j_1 \cdots j_m}=\frac{\sum_{\sigma \in \Sigma\left(j_1, \cdots, j_m\right)} {\mathbf{A}}_{\sigma\left(j_1, \cdots, j_m\right)}}{\left|\Sigma\left(j_1, \cdots, j_m\right)\right|} 
\end{equation}
\end{lem}

Based on Lemmas \ref{lem2}-\ref{lem3}, the following theorem can be derived.

\begin{thm} \label{thm:symO}
Higher-order network system \eqref{simplified} is globally observable at an initial state $\sigma \in {\mathbf{R}}^n$ if and only if condition (a) or (b) in Lemma \ref{lem2} holds with
$\mathcal{J}=J_N$. $\mathcal{J}$ is
\begin{align}
\mathcal{J}= & \langle   \{{\Pi}_{i=1}^{\tilde k} (i(k-1)+1) E_{\tilde{k}}  \xi^{\tilde{k}(k-1)+1}\nn\\
&-{\Pi}_{i=1}^{\tilde k} (i(k-1)+1) E_{\tilde{k}}  \sigma^{\tilde{k}(k-1)+1}|\tilde k\in \ell_N \}\rangle 
\end{align}
and $E_{\tilde{k}}:=\operatorname{sym}(\cdots\operatorname{sym}(\mathbf C\circ \mathbf A)\cdots\circ \mathbf A)$.
 \end{thm}

 \begin{proof}    
By simple calculations, we can confirm that if $\mathbf A$ is super-symmetric, then the derivative of $f(x)=\mathbf Ax^m$ is $\frac{df(x)}{dx}=m\mathbf Ax^{m-1}$.
Define $\tau:=\mathbf{A}{x}^{k-1}$ and $v:=\mathbf{C}{x}^{k}$. 
It holds that 
\begin{align}
L_\tau v
=&\frac{\partial (\mathbf{C}{x}^{k})}{\partial x} \mathbf{A}{x}^{k-1} \nn\\
=&\sum_{i=1}^{n}(k\mathbf Cx^{k-1})_i(\mathbf Ax^{k-1})_i \nn\\
=&k\sum_{i=1}^{n}\left(\sum_{i_1,\cdots, i_{k-1}=1}^{n}\mathbf C_{i_1\cdots i_{k-1}i}x_{i_1}\cdots x_{i_{k-1}}\right)\nn\\
&\times \left(\sum_{j_1,\cdots, j_{k-1}=1}^n \mathbf A_{j_1\cdots j_{k-1}i}x_{j_1}\cdots x_{j_{k-1}}\right)\nn
\end{align}

Define a new type of tensor-tensor multiplication operator as ``$\circ$" such that 
\begin{align}
(X\circ Y)_{i_1\cdots i_{k-1} j_1\cdots j_{k-1}}:=\sum_{i=1}^{n}X_{i_1\cdots i_{k-1}i}Y_{j_1\cdots j_{k-1}i} \label{multi}
\end{align}
where $X$ and $Y$ are $n$-dimensional tensors, i.e., $i_1,\cdots, i_{k-1}, j_1,\cdots, j_{k-1},i\in \{1,\cdots, n\}$.

Therefore, according to \eqref{multi}, it holds that
\begin{align}
L_\tau v
=&k\sum_{i=1}^{n}  \Bigg(\sum_{i_1,\cdots, i_{k-1}=1}^{n}\sum_{j_1,\cdots, j_{k-1}=1}^n E_{i_1\cdots i_{k-1}i j_1\cdots j_{k-1}} \nn\\
&\times x_{i_1}\cdots x_{i_{k-1}}x_{j_1}\cdots x_{j_{k-1}}\Bigg) \nn\\
=& kE_1 x^{2k-2} \nn
\end{align}
where 
\begin{align}
{(E_1)}_{i_1\cdots i_{k-1} j_1\cdots j_{k-1}}:=\sum_{i=1}^{n}\mathbf C_{i_1\cdots i_{k-1}i}\mathbf A_{j_1\cdots j_{k-1}i} \nn
\end{align}

According to Lemmas 2-3, $E_1$ can be symmetrized by $\operatorname{sym}(E_1)$.

Furthermore, according to Lemma \ref{lem1} and \eqref{multi}, one can obtain 
\begin{align}
L_\tau L_\tau v= k(2k-2)E_2 x^{3k-3} 
\end{align}
where $E_2:=\operatorname{sym}(E_1)\circ \mathbf A$.


Therefore, by induction, one can conclude that 
\begin{align}
\overbrace{{L_\tau} \cdots {L_\tau}}^{\tilde{k} \ \mathbf{times}} v=
k{\Pi}_{i=2}^{\tilde k}( i(k-1) )E_{\tilde{k}}  x^{\tilde{k}(k-1)} \nn
\end{align}
where $E_{\tilde{k}} :=\operatorname{sym}(E_{\tilde{k}-1})\circ \mathbf A$. 


Define $E_0:=C$ and $\ell_t:=\{0,\cdots, t\}$. Thus, for simplified system \eqref{simplified}, consider the initial state pair $(\xi,\eta)$. 
$J_{t}$ can be derived as
\begin{align}
J_{t}= & \langle   \{k{\Pi}_{i=2}^{\tilde k} (i(k-1)) E_{\tilde{k}}  \xi^{\tilde{k}(k-1)}\nn\\
&-k{\Pi}_{i=2}^{\tilde k}(i(k-1)) E_{\tilde{k}}  \eta^{\tilde{k}(k-1)}|\tilde k\in \ell_t \}\rangle 
\end{align}

Fix $\eta$ as an initial state $\sigma$. According to Lemma \ref{lem2}, there exists $N$ such that $v(\mathcal{J})=\{\sigma\}$
\begin{align}
\mathcal{J}= & \langle   \{{\Pi}_{i=1}^{\tilde k} (i(k-1)+1) E_{\tilde{k}}  \xi^{\tilde{k}(k-1)+1}\nn\\
&-{\Pi}_{i=1}^{\tilde k} (i(k-1)+1) E_{\tilde{k}}  \sigma^{\tilde{k}(k-1)+1}|\tilde k\in \ell_N \}\rangle 
\end{align}
and $E_{\tilde{k}}:=\operatorname{sym}(\cdots\operatorname{sym}(\mathbf C\circ \mathbf A)\cdots\circ \mathbf A)$.
 \end{proof}

\begin{rem} \label{rem5} This tensor-based formulation provides a clearer and more structured approach, simplifying the computation process and making the observability conditions more accessible. For system (\ref{aa2}), the definition of ideal $\mathcal{J}$  can be obtained with the same process as the simplified system \eqref{simplified} via the multiplication \eqref{multi} and Lemmas 2-3. Specifically, by replacing $\mathbf{A}_k$ with $\sum_{k=1}^{c}\mathbf{A}_k{x}^{k-1}$ and $\sum_{k=1}^{c}\mathbf{B}_{k,j}{x}^{k-1}$ for $j=1,\cdots,m$, and replacing $\mathbf{C}_{i,k}$ with $\sum_{k=1}^{c}\mathbf{C}_{k}{x}^{k}$ and $\sum_{k=1}^{c}\mathbf{D}_{i,k,l}{x}^{k}$ for $l=1,\cdots,w_i$, it is evident that the conclusion of Theorem \ref{thm:symO} can also be extended to more general forms of the system with the same procedure.
 \end{rem}

{\color{blue}
In the following, to further investigate the association between non-uniform hypergraph structure and the observability criteria, consider higher-order networks under non-uniform directed hypergraphs whose dynamics are
\begin{align}\label{eq:sys_multiA}
 \left\{\begin{aligned} 
\dot{x} &= \sum_{k=2}^{c}\mathbf{A}_k x^{k-1}  \\
y_i &= \sum_{k=1}^{c}\mathbf{C}_{i,k} x^{k}
\end{aligned}\right.
\end{align}
where $i=1,\ldots,q$, $\mathbf{A}_k\in\mathbb{R}^{n\times n^{k-1}}$ 
and $\mathbf{C}_{i,k}\in\mathbb{R}^{1\times n^{k}}$ 
are \emph{non-symmetric} tensors representing
the $k$-th order dynamic and output couplings, respectively.
Each $\mathbf{A}_k$ defines a directed $k$-uniform hyperedge family,
and each $\mathbf{C}_{i,k}$ corresponds to a directed output hyperedge
anchored at the $i$-th sensor node.

\begin{thm} \label{thm:generalO}
Higher-order network system \eqref{eq:sys_multiA} is globally observable at an initial state $\sigma \in {\mathbf{R}}^n$ if and only if condition (a) or (b) in Lemma \ref{lem2} holds with
$\mathcal{J}=J_N$. $\mathcal{J}$ is
\begin{align}
\!\!\mathcal{J}\!=\!\!\langle
\{
E_r(\xi^{\otimes d_r}&)-E_r(\sigma^{\otimes d_r})
\mid 
E_r= \nn
\\&\mathbf{C}^{(i,k_1)}_{(s_1)}\!\circ\!
\mathbf{A}_{k_2}\!\circ\!\cdots\!\circ\!\mathbf{A}_{k_r},
~
r=0,\ldots,N
\}
\rangle
\label{eq:J_multiA}
\end{align}
\end{thm}

\begin{proof}
According to \eqref{eq:sys_multiA},
output $y_i$ satisfies
\[
y_i = \sum_{k=1}^{c}
\sum_{j_1,\ldots,j_k}
C_{i,k}(j_1,\ldots,j_k)
x_{j_1}\cdots x_{j_k}.
\]
Taking the Lie derivative of $y_i$ along $\dot{x}=\sum_{m=2}^{c}\mathbf{A}_m x^{m-1}$ gives
\begin{align}
L_f y_i
=& \sum_{k=1}^{c}\sum_{s=1}^{k}
   \sum_{j_1,\ldots,j_k}
   C_{i,k}(j_1,\ldots,j_k)
   x_{j_1}\!\cdots\! \dot{x}_{j_s}\!\cdots\! x_{j_k} \notag\\
=& \sum_{k=1}^{c}\sum_{s=1}^{k}
   \sum_{j_1,\ldots,j_k}\sum_{m=2}^{c}
   \sum_{p_1,\ldots,p_{m-1}}
   \!\!\!\!\!\!C_{i,k}(j_1,\ldots,j_k)
   A_{j_s p_1\ldots p_{m-1}}\nn\\
  &\times x_{j_1}\!\cdots\!x_{j_{s-1}}
   x_{p_1}\!\cdots\!x_{p_{m-1}}
   x_{j_{s+1}}\!\cdots\!x_{j_k}
\end{align}
For each index $s$, define an \emph{index-shifted tensor} as
\begin{align}
\mathbf{C}^{(i,k)}_{(s)}(j_1,\ldots,j_{s-1},&t,j_{s+1},\ldots,j_k)
:=\nn\\
&C_{i,k}(j_1,\ldots,j_{s-1},t,j_{s+1},\ldots,j_k)
\end{align}
Thus, the compacted form is
\begin{equation}
L_f y_i
= \sum_{k=1}^{c}\sum_{s=1}^{k}\sum_{m=2}^{c}
  \big(\mathbf{C}^{(i,k)}_{(s)}\circ \mathbf{A}_m\big)
  x^{k+m-2}
\label{eq:Lfy_multiA}
\end{equation}
where ``$\circ$'' denotes mode contraction preserving the index direction.

By iterating~\eqref{eq:Lfy_multiA},
we obtain for the $r$-th Lie derivative
\begin{equation}
L_f^{r} y_i
= \sum_{\substack{
(k_1,\ldots,k_r)\in\{2,\ldots,c\}^r\\
(s_1,\ldots,s_r)\in\mathcal{I}}}
  \mathbf{C}^{(i,k_1)}_{(s_1)}\circ
  \mathbf{A}_{k_2}\circ\cdots\circ\mathbf{A}_{k_r}\,
  x^{d_r}
\label{eq:Lfr_multiA}
\end{equation}
with $d_r = k_1+\sum_{j=2}^{r}(k_j-1)$
and $\mathcal{I}$ collecting admissible index positions. Hence, the ideal
\begin{align}
\!\!J=\!\!\langle
\{
E_r(\xi^{\otimes d_r}&)-E_r(\eta^{\otimes d_r})
\mid 
E_r= \nn
\\&\mathbf{C}^{(i,k_1)}_{(s_1)}\!\circ\!
\mathbf{A}_{k_2}\!\circ\!\cdots\!\circ\!\mathbf{A}_{k_r},
~
r=0,\ldots,N
\}
\rangle
\label{al:J_multiA}
\end{align}
contains finite generators since $J_N=J_{N+1}=J$ by
Hilbert’s basis theorem. 
\end{proof}

\begin{rem} \textit{(Structural Interpretation of the algebraic observability criteria on Hypergraphs)}
For a general higher-order network system under non-uniform directed hypergraphs, Theorem \ref{thm:generalO} clearly indicates that each item corresponds to a directed higher-order propagation path.
Each tensor $\mathbf{A}_k$ defines a family of directed hyperedges $(i_0,i_1,\ldots,i_{k-1})$ with $k$ nodes in system dynamics hypergraph $\mathcal{G}_f$,
and each $\mathbf{C}_{i,k}$ defines a $k$-node hyperedge in output hypergraph $\mathcal{G}_y$ connecting
state nodes to output node~$i$.
The ordered tensor contractions in~\eqref{eq:Lfr_multiA}
correspond to \emph{directed paths}
on this hypergraph.
\textit{Each ordered combination
$(\mathbf{C}^{(i,k_1)}_{(s_1)},\mathbf{A}_{k_2},\ldots,\mathbf{A}_{k_r})$
represents a distinct directed higher-order propagation path
from state nodes to the output $y_i$.}
Mathematically, $\mathbf{C}^{(i,k_1)}_{(s_1)}\!\circ\!
\mathbf{A}_{k_2}\circ\cdots\circ\mathbf{A}_{k_r}$ are equivalent to the path algebra of hypergraphs.
\end{rem}

\section{Insights on Algebraic Observability and Structural Designability}
\label{sec:insights}

In this section, based on the algebraic observability criteria established in Theorems~\ref{thm:symO}-\ref{thm:generalO}, we focus on creating the connection between algebraic global observability and designability of both structure and weight for hypergraphs.  
Several important insights are obtained in the following.


\subsection{Algebraic Criterion Overcomes Local Rank Deficiency}

For conventional nonlinear systems, the \emph{observability rank condition} only evaluates local information around an initial state. Once the observability matrix is rank-deficient at a point, this method cannot determine whether the system is observable or not. In contrast, our algebraic observability criterion completely breaks this limitation and exhibits a natural advantage when dealing with such systems. Taking the general higher-order network system \eqref{eq:sys_multiA} as an example, the following is an explanation of this insight. 

Under the observability rank-deficient case:
\begin{itemize}
\item
The ideal-based observability conditions in Theorems~\ref{thm:symO}--\ref{thm:generalO} remain not only sufficient but also necessary.
This ensures that global observability of a higher-order network system can still be conclusively determined even when the classical rank test fails.
\item
Rank deficiency implies that the polynomial ideal chain $J_0, J_1,\cdots, J_N$ terminates early, with $N<nq$ and possibly even $N=1$. Hence, global observability can be checked in only a few steps—or even a single step. For instance, if $L_f y=\sum_{k=1}^{c}\sum_{s=1}^{k}\sum_{m=2}^{c}
  \big(\mathbf{C}^{(k)}_{(s)}\circ \mathbf{A}_m\big)
  x^{k+m-2}=0$, the only remaining test is whether $\sum_{k=1}^{c}\mathbf{C}_{k} \xi^{k}- \sum_{k=1}^{c}\mathbf{C}_{k} \sigma^{k}=0$ holds when $\xi=\sigma$.
\end{itemize}


\textit{Example.}
Consider the observability at the initial state $\sigma=0$ for higher-order dynamics $\dot x= \mathbf{A}_3 x^2, y=\mathbf{C}_2x^2$ with
$(\mathbf{A}_3)_{231}=-\frac{1}{2}$, $(\mathbf{A}_3)_{132}=1$, $(\mathbf{A}_3)_{123}=-\frac{1}{2}$, and any other $(\mathbf{A}_3)_{i_1i_2i_3}=0$. For output hypergraph, $(\mathbf{C}_2)_{11}=(\mathbf{C}_2)_{22}=(\mathbf{C}_2)_{33}=1$, and any other $(\mathbf{C}_2)_{ij}=0$.

\begin{figure}[h]
        \centering
        \includegraphics[height=5cm]{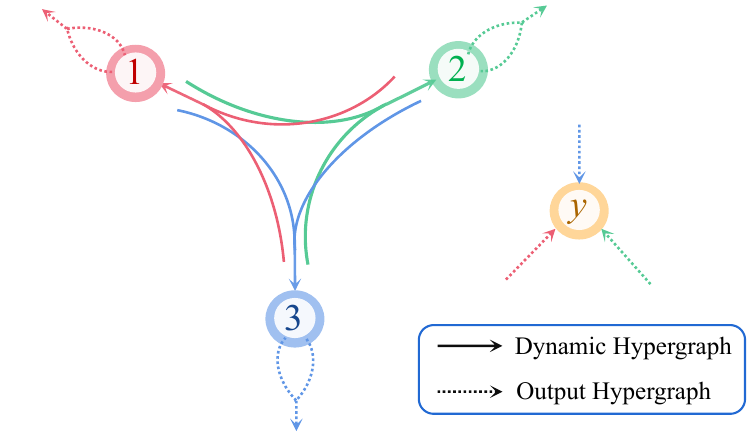}
        \caption{Example of a higher-order network with deficient observability matrix rank}
        \label{example1}
\end{figure}

Since $L_f y=\sum_{k=1}^{c}\sum_{s=1}^{k}\sum_{m=2}^{c}
  \big(\mathbf{C}^{(k)}_{(s)}\circ \mathbf{A}_m\big)
  x^{k+m-2}=2x_1(-\frac{1}{2}x_2x_3)+2x_2x_1x_3+2x_3(-\frac{1}{2}x_1x_2)=0$, the classical observability matrix has rank~1 for all $x$, so the rank test is inconclusive.  
However, according to Theorem \ref{thm:generalO}, the algebraic criterion gives $N=1$ and
$J_1={\xi}_1^2+{\xi}_2^2+{\xi}_3^2-{\eta}_1^2-{\eta}_2^2-{\eta}_3^2$. Then $V(\mathcal{J})=V({\xi}_1^2+{\xi}_2^2+{\xi}_3^2)=0$ leads to ${\xi}=0$. Since $V(\ell)=0$ leads to ${\xi}=0$,
$V(\mathcal{J})=V(\ell)$. This implies that $\sigma=0$ is the only global indistinguishable state.  
Hence, despite the local rank deficiency, the system is globally observable at $\sigma=0$.  
This shows that the proposed algebraic test can immediately identify global observability in observable rank-deficient systems where conventional methods fail.

\subsection{Hyperedge-Level Co-Design Enables Algebraic Reconstruction}

Although the observability condition is algebraic in nature, it provides a direct and constructive rule for system design. Consider higher-order network system \eqref{eq:sys_multiA}.
Since the tensors $\{\mathbf{A}_k\}$ and $\{\mathbf{C}_k\}$, which are hyperedges in dynamic and output hypergraphs, together determine the generators of the polynomial ideal $\mathcal{J}$, adjusting these tensors changes the algebraic dependency among group-wise states.  
Therefore, by slightly modifying $\mathbf{A}_k$ or $\mathbf{C}_k$, one can remove algebraic redundancy and recover global observability.  
This makes the algebraic criterion not only an analysis tool but also a practical \emph{design rule}: in a few algebraic steps, one can guide how to perturb or reconfigure the tensors to achieve global observability.

\textit{{Example.}}
Consider a higher-order network system \eqref{eq:sys_multiA} with ${(\mathbf{A}_{3})}_{123}=1$ ${(\mathbf{C}_1)}_{3}=1$. The ideal is $\mathcal{J}=\langle \xi_3-\sigma_3,\ \xi_1\xi_2-\sigma_1\sigma_2\rangle$.
In this case, $V(\mathcal{J})=\{\xi|\xi_3=\sigma_3, \xi_1\xi_2=\sigma_1\sigma_2\}$. Since $V(\ell)=\{\sigma\}$, $V(\mathcal{J})\not=V(\ell)$. 

Now we try to reconstruct the system to recover the observability.
we analyze the original ideal and augment the output by adding a new output ${(\mathbf{C}_1)}_{1}=1$ such that $y=[x_3,\,x_1]^\top$, the new ideal becomes
\[
\mathcal{J}=\langle \xi_3-\eta_3,\ \xi_1-\eta_1,\ \xi_1\xi_2-\eta_1\eta_2\rangle.
\]
Now, all variables $\xi_1,\xi_2,\xi_3$ appear in the generators, and $V(\mathcal{J})=V(\ell)$ holds.  
Thus, the system becomes globally observable simply by adjusting output $\mathbf{C}$ (adding one component to the output), without changing the overall topology or recomputing ranks.  
Similarly, one could perturb a single coefficient in $\mathbf{A}_2$ to break the algebraic symmetry and re-achieve observability.  
This illustrates that the co-design of $\mathbf{A}_k$ and $\mathbf{C}_k$ provides an algebraically guided reconstruction path from unobservable to observable configurations.

\subsection{Incremental Output Design via Lie-Derivative Vanishing}

Once the system dynamics are fixed, another insight is that global observability can also be achieved through an \emph{incremental Lie-derivative vanishing design} by appropriately choosing the output hypergraph. \textit{Our goal is to choose $\mathbf{C}_k$ so that the algebraic observability test is highly efficient}. Thus, we design the following algorithm to illustrate the design process. }

\begin{algorithm}[htbp]
{\color{blue}
\caption{Incremental Design of Output Hypergraph to Achieve Observability}
\KwIn{fixed $\{\sum_{k=2}^c\mathbf{A}_k\}$; maximal cardinality of output hyperedges $d_{\max}$; maximum allowable number of sensors $p$}
\KwOut{ $\{\mathbf{C}_{i,k}\}, i\leq p$ with $V(\mathcal J)=V(\ell)$ if found}
\textbf{Init:} $d\gets 1$; choose monomial basis $\mathcal M_d$.\;
\Repeat{nontrivial $\Gamma$ found or $d>d_{\max}$}{
  Parametrize $y_i=\sum_{k=1}^{c}\mathbf{C}_{i,k}x^k=\sum_{m\in\mathcal M_d}\gamma_{i,m}m(x)$; collect $\Gamma$\;
  Enforce $L_f y_i\equiv 0$, expand in $\mathcal M_d$ $\mathcal  M_d(\sum_{k=2}^c\mathbf{A}_k)\,\Gamma=0$.\;
  \lIf{$\mathrm{Null}(\mathcal M_d)=\{0\}$}{ $d\gets d+1$; enlarge $\mathcal M_d$ }
}
\If{$d>d_{\max}$ \textbf{ and } no $\Gamma$}{\textbf{go to} Higher-Order Step}
Choose $\Gamma\in \mathrm{Null}(M_d)$ with (coverage \& asymmetric supports).\;
Build $J_0:=y(\xi)-y(\eta)$ with $\eta=\sigma$; compute Gröbner basis and apply radical/variety test.\;
\While{$V(\mathcal J)\neq V(\ell)$}{
  Minimally augment $y$ (within $p$ or by increasing $d$ while keeping $L_f y_i\equiv 0$);\;
  Recompute $J_0$ and test again.\;
  \If{budget exhausted}{\textbf{go to} Higher-Order Step}
}
\textbf{return} designed $\{\mathbf{C}_{i,k}\}$.\;

\BlankLine
\textbf{Higher-Order Step:} relax to sparse $y$ with $L_f y_i\not\equiv 0$ and $L_f^2 y_i\equiv0$, or further to $L_f^r y_i\equiv0$ for $r\ge 3$, 
build the extended ideal chain $J_1,\ldots,J_r$ and test stabilization.\;
\lIf{$V(\mathcal J)=V(\ell)$}{\textbf{return} $\{\mathbf{C}_{i,k}\}$}
\lElse{\textbf{return} failure under current budget}}
\end{algorithm}

{\color{blue}
The \textit{design principle} is to construct the output $\mathbf{C}_{i,k}$ for a known higher-order network system such that it (i) minimally covers all state coordinates present in the generators of the dynamics and (ii) breaks all algebraic indistinguishability through a designed hyperedge support. A convenient starting point is to enforce the first-order Lie-derivative vanishing constraint, $L_f y(x) \equiv \sum_{k=1}^{c}\sum_{s=1}^{k}\sum_{m=2}^{c}
  \big(\mathbf{C}^{(k)}_{(s)}\circ \mathbf{A}_m\big)
  x^{k+m-2}\equiv 0$, which consequently forces all higher-order Lie derivatives to vanish. In this case, the observability problem reduces to verifying the zero-set equivalence $y(x)=\sum_{k=1}^{c}\mathbf{C}_{i,k} x^k=0$ has only one solution (i.e., $V(\mathcal{J}) = V(\ell)$), which can be checked via a Gröbner basis computation on the generators of the ideal $\mathcal{J}_0 = \langle y(\xi) - y(\eta) \rangle$. The search for $y$ is thus narrowed to outputs whose hyperedge tensors belong to the Lie-derivative vanishing nullspace and whose monomial supports ensure the variety of algebraic ideal reduces to the single point $\sigma$, which is the initial state.

\textit{Example.} 
Consider the higher-order dynamics
\[
\dot{x}_1 = -\tfrac12 x_2 x_3,\qquad
\dot{x}_2 =  x_1 x_3,\qquad
\dot{x}_3 = -\tfrac12 x_1 x_2 .
\]
According to Algorithm 1, start with the degree-$1$ monomial basis
$\mathcal M_1=\{x_1,x_2,x_3\}$. The output is parametrized as
$y=\gamma_1 x_1+\gamma_2 x_2+\gamma_3 x_3$.
Enforcing $L_f y\equiv 0$ yields $\gamma_1=\gamma_2=\gamma_3=0$, so no feasible
output exists at degree~$1$.

Enlarging to the degree-$2$ basis
\[
\mathcal M_2=\{x_1^2,\,x_2^2,\,x_3^2,\,x_1x_2,\,x_1x_3,\,x_2x_3\},
\]
and imposing $L_f y\equiv 0$ yields nontrivial solutions. A feasible choice is
\[
y = x_1^2 + x_2^2 + x_3^2 .
\]
The indistinguishability ideal
\[
J_0 = \xi_1^2+\xi_2^2+\xi_3^2 - (\eta_1^2+\eta_2^2+\eta_3^2)
\]
satisfies $V(\mathcal J_0)=\{0\}$ when $\sigma=0$, so the system is globally observable at the initial state $\sigma=0$.


\section{From Algebraic Observability to Structural Observability}
For higher-order network system \eqref{eq:sys_multiA}, while the algebraic observability results above characterize global identifiability in terms of Lie–ideal closure and polynomial variety equalities, these conditions inherently depend on the specific numerical values of the tensors $\mathbf{A}_k$ and $\mathbf{C}_{i,k}$ which represents the interactions in system dynamics and output hypergraphs $\mathcal{G}_f$ and $\mathcal{G}_y$, respectively.
However, in many higher-order networked systems, the coupling structure—captured by the zero–nonzero pattern of the tensors—is determined by the underlying communication, interaction, or physical topology, whereas the actual weights may vary with operating conditions, control objectives, or learning processes.
This motivates an analysis of observability that does not rely on precise parameter values but instead depends only on the interaction structure.
In other words, structural observability for higher-order network systems is studied in this section.

\subsection{Structural observability criterion for higher-order networks}

Before going on, we define the structural global observability for hypergaphs.

\begin{defn} \label{def:so}
 Consider a higher-order dynamical system defined on a dynamic hypergraph $\mathcal{G}_f=(\mathcal{V},\mathcal{E}_f)$ and an output hypergraph $\mathcal{G}_y =(\mathcal{V},\mathcal{E}_y)$. The system is said to be \textit{structurally globally observable} if for almost all admissible parameters of the tensors $\mathbf{A}_k$ and $\mathbf{C}_{i,k}$ consistent with the nonzero structure of $\mathcal{G}_f$ and $\mathcal{G}_y$ such that $V(\mathcal{J})=V(\ell)$ at any initial state $\sigma\in \mathbf{R}^n$.
\end{defn}


Consider the non-uniform directed system \eqref{eq:sys_multiA}
For two initial states $\xi,\eta\in\mathbb{R}^n$, define the polynomial differences
\begin{align}\label{Phi}
\Phi_s(\xi,\eta)= L_f^{s}y(\xi)-L_f^{s}y(\eta),\qquad s\ge 0,
\end{align}
and the ideal chain $J_0=\langle \Phi_0\rangle$, $J_{s+1}=J_s+\langle \Phi_{s+1}\rangle$ with $J_N=J_{N+1}=J$ for some finite $N$.

Recall that for directed hyperedge $\mathcal{E}:=(i_1,i_2,\cdots, i_n)\rightarrow j$, $i_1,i_2,\cdots, i_n\in \mathrm{head}(\mathcal{E})$, while $j\in \mathrm{tail}(\mathcal{E})$. 
In the following, we define the observational diameter as the observation propagation depth in the higher-order networks with both dynamic and output hypergraphs. 

\begin{defn}
Define the observational closure
\begin{align}
&\mathcal{R}_o^{(0)}:= \mathrm{head}(\mathcal{E}_y)\nn\\
&\mathcal{R}_o^{(t+1)}\!\!:= \mathcal{R}_o^{(t)} \,\cup\, 
\left\{ \mathrm{head}(\mathcal{E}_f^{(k)})|{\mathcal{E}_f^{(k)}\in\mathcal{E}_f:\, \mathrm{tail}(\mathcal{E}_f^{(k)})\in \mathcal{R}_o^{(t)}}\right\} \nn
\end{align}
where $\mathcal{E}^{(k)}$ is the all $k$-th order hyperedges and $k=2,\cdots,c$. Define the \emph{observational diameter} as $T:=\mathrm{inf}\{t:\mathcal{R}_o^{(t)}=\mathcal{V}\}$. If there exists no $t$ that leads to $\mathcal{R}_o^{(t)}=\mathcal{V}$, let $T=\infty$.
\end{defn}

According to the definition, the following lemma of observational diameter can be obtained directly.

\begin{lem}
\label{lem:T_equals_N}
Consider the system \eqref{eq:sys_multiA}. Partition the node set $\mathcal V$ into layers
$L_0:=\mathcal R_o^{(0)}$,
$L_{t+1}:=\mathcal R_o^{(t+1)}\!\setminus\!\mathcal R_o^{(t)}$.
For each node $j$, define its backward distance
$d(j):=\min\{t:j\in\mathcal R_o^{(t)}\}$. If there exists $t$ such that $\mathcal{R}_o^{(t)}=\mathcal{V}$, it holds that $T=\max_jd(j)$. Otherwise, $T=\infty$.
\end{lem}

%

\begin{lem}
\label{lem:coverage_variables}
Let $\mathcal{R}_o^{(\infty)}:=\bigcup_{t\ge 0}\mathcal{R}_o^{(t)}$.
For any node $v_j\in\mathcal{V}$ with state coordinate $x_j$:
\noindent If there exists node $j\notin \mathcal{R}_o^{(\infty)}, j\in\{1,2,\cdots, q\}$, then $x_j$ does not appear in any generator of $J$, and there exists a one-parameter family of indistinguishable initial states varying $x_j$.  
Conversely, under generic parameters, if $j\in \mathcal{R}_o^{(t)}$ for some finite $t$, then $x_j$ appears in at least one generator of $J_i$, $i\leq N$.
\end{lem}

\begin{proof}
If node $j$ is never reached, no directed path from node $j$ to any output hyperedge exists. Thus no contraction sequence in any $L_f^s y$ can contain $x_j$, so $\partial \Phi_s/\partial x_j\equiv 0$ for all $s$, and $x_j$ is absent from all generators. Varying $x_j$ while fixing other coordinates leaves all outputs and their Lie derivatives unchanged, creating a continuum of indistinguishable initial states.
If node $j\in \mathcal{R}_o^{(t)}$, there exists a directed sequence of hyperedges of length $t$ carrying $x_j$ to an output. Under generic parameters, the corresponding ordered contraction yields a nonzero monomial containing $x_j$ in $L_f^{t}y$, hence $x_j$ appears in some generator of $J_i$.
\end{proof}

\begin{rem}
In Lemma \ref{lem:T_equals_N}, the integer $T<\infty$ means that all nodes become observable after $T$ layers of the backward hyperedge propagation. Lemma \ref{lem:coverage_variables} illustrates that the reachability of node $j\in\{1, 2,\cdots, q\}$ from the output is equivalence to the existence of the state $x_j$ in generators $J_i,i\leq N$.
\end{rem}

Based on the above two lemmas, the following theorem of structural observability criterion for higher-order networks can be derived.

\begin{thm} \label{thm-so}
Consider the higher-order networked system \eqref{eq:sys_multiA}, whose dynamic and output structures are described by the directed hypergraphs 
$\mathcal{G}_f$ and $\mathcal{G}_y$, respectively.
The system is \emph{structurally globally observable} if the following two conditions hold:

\begin{itemize}
    \item[] \textbf{C1) Topological reachability. } The observational diameter is finite, i.e., $T<\infty$,
    \item[] \textbf{C2) Symmetry breaking.} The coupled dynamic–output hypergraph is structurally asymmetric, i.e.,
    $\mathrm{Aut}(\mathcal{G}_f,\mathcal{G}_y)=\{\mathrm{id}\}$,
\end{itemize}
where $\mathrm{Aut}(\mathcal{G}_f,\mathcal{G}_y)$ denotes the automorphism group of the coupled hypergraph, and $\mathrm{id}$ denotes the identity permutation.
\end{thm}

\begin{proof}
According to Lemma \ref{lem:T_equals_N}, 
Condition C1) ensures that $\max_j d(j)=T<\infty$.
By Lemma~\ref{lem:coverage_variables}, each coordinate $x_j$
first appears in some generator $\Phi_{s_j}$ with
$s_j\le d(j)$, and this appearance is realized by a
\emph{witness monomial}
$m_j(\delta)$ where $\delta:=\xi-\eta$
whose support encodes at least one backward hyperpath from an
output node to node $j$.  Hence
\begin{equation}  \label{eq:Nvar=T-en}
T=\max_j d(j)\geq \max_j s_j
\end{equation}

Define a monomial order ``$\prec$" consistent with the layer structure:
variables belonging to lower layers are larger,
and within the same layer a lexicographic order is used.
For each $j$, choose in $\Phi_{s_j}$ the witness monomial $m_j(\delta)$
that contains $\delta_j$ with the smallest possible exponent.
Under generic independent parameters,
the coefficient of $m_j$ is nonzero, and because
$\mathrm{Aut}(\mathcal G_f,\mathcal G_y)=\{\mathrm{id}\}$,
no two distinct nodes share identical sets of backward hyperpath signatures, where backward hyperpath signature of node $j$ as
the set of all directed higher-order hyperedge chains that start from any output hyperedge and terminate at $j$. $\mathrm{Aut}(\mathcal{G}_f, \mathcal{G}_y)=\{\pi: \mathcal{V}\rightarrow \mathcal{V}|\pi(\mathcal{E}_f)=\mathcal{E}_f, \pi(\mathcal{E}_y)=\mathcal{E}_y\}$. Algebraically, the monomial support sets will not change for $J_M, \forall M\geq \max_j{s_j}$. Thus, under generic and independent parameters, $V(J_M)=V(J_{M+1}),\forall M\geq \max_j{s_j}$. According to \eqref{eq:Nvar=T-en}, one has $V(J_{\max_j{s_j}})=\cdots=V(J_{T})=V(J_{T+1})=\cdots$.

Substituting any known initial state $\sigma$ into $\eta$, define $\delta:=\xi-\sigma$ and $\mathcal{J}_T(\delta):=J_T(\sigma+\delta, \delta):=\langle \Phi_0(\delta),\Phi_1(\delta),\cdots, \Phi_T(\delta)\rangle$ where $\Phi$ is defined in \eqref{Phi}.
Therefore, the leading monomials of the generators $\Phi_s, \forall s\leq T$ include $\{\delta_1^{\alpha_1},\dots,\delta_n^{\alpha_n}\}$
for some integers $\alpha_j\ge1$, i.e., the initial ideal
\begin{equation}\label{eq:init-ideal}
\mathrm{in}_\prec(\mathcal{J}_T)\ \supseteq\
\big\langle\delta_1^{\alpha_1},\delta_2^{\alpha_2},\dots,\delta_n^{\alpha_n}\big\rangle.
\end{equation}

From~\eqref{eq:init-ideal},
the radical of the initial ideal satisfies
$\sqrt{\mathrm{in}_\prec(\mathcal{J}_T)}\supseteq\langle\delta_1,\dots,\delta_n\rangle$.
According to Theorem 3.3.4 in \cite{Cox2}, the standard result from Gröbner-basis theory holds
\[
\sqrt{\mathrm{in}_\prec(\mathcal{J}_T)}\subseteq
\mathrm{in}_\prec\big(\sqrt{\mathcal{J}_T}\big)\subseteq \sqrt{\mathcal{J}_T}
\]
which implies
\[
\langle\delta_1,\dots,\delta_n\rangle
\subseteq
\sqrt{\mathcal{J}_T}
\]
Since each $\Phi_s$ vanishes on the diagonal $\xi=\eta$, $V(\langle\delta_1,\dots,\delta_n\rangle)\subseteq V(\sqrt{\mathcal{J}_T})$, which leads to
$\sqrt{\mathcal{J}_T}\subseteq\langle\delta_1,\dots,\delta_n\rangle$.
Then it follows:
\[
\sqrt{\mathcal{J}_T}=\langle\delta_1,\dots,\delta_n\rangle.
\]
Consequently, $V(\mathcal{J}_T)=V(\ell)$. under independent independent parameters, it combines with  $V(J_{\max_j{s_j}})=\cdots=V(J_{T})=V(J_{T+1})=\cdots$ follows $V(\mathcal{J})=V(\ell)$. This establishes global structural observability according to Definition \ref{def:so}. 
\end{proof}

\begin{rem}
   Condition C1) means that every node can reach an output node through a finite sequence of backward-directed hyperedges, while Condition C2) eliminates any structural symmetry that could cause indistinguishable state directions. \textit{Intuitively, the first condition  guarantees that all the nodes can be ``reachable" from the system outputs, while the second condition ensures that they are distinguishable.} These two straightforward conditions enable a highly effective observability verification process from the higher-order topology perspective.
\end{rem}

\subsection{Illustrative Examples of Structural Observability}

To illustrate the structural observability criterion established in Theorem~\ref{thm-so}, 
we present two small-scale higher-order systems showing the effect of structural symmetry 
on observability.

\textit{{Example.}}
Consider the higher-order network system with dynamics
\begin{equation}
    \dot{x}_3 = a x_1 x_2,\qquad \dot{x}_1=\dot{x}_2=0,\qquad y=x_3
\end{equation}
The dynamic hyperedge is $(1,2)\!\to\!3$, and the output hyperedge is $3\!\to\!y$. The hypergraph structure is presented in Fig. \ref{symmetric} (a).
The backward closure gives $L_0=\{3\}$ and $L_1=\{1,2\}$, hence $T=1$. 
Since swapping nodes $1\leftrightarrow 2$ leaves all hyperedges invariant, 
the automorphism group is nontrivial, i.e., $\mathrm{Aut}\neq\{\mathrm{id}\}$. Thus the hypergraph structure 
violating Condition~C2. According to Theorem \ref{thm-so}, the system is not structurally observable.

Now we verify our results algebraically. For two initial states $\xi,\eta$,
\[
\Phi_0(\xi,\eta)=\xi_3-\eta_3,\qquad
\Phi_1(\xi,\eta)=a(\xi_1\xi_2-\eta_1\eta_2)
\]
Setting $\eta=\sigma$ and $\delta=\xi-\sigma$ yields
\[
\Psi_0(\delta)=\delta_3,\qquad
\Psi_1(\delta)=a(\sigma_1\delta_2+\sigma_2\delta_1+\delta_1\delta_2)
\]
If $\sigma_1=\sigma_2=0$, then $\Psi_1(\delta)=a\,\delta_1\delta_2$, so
\[
\mathcal{J}_1=\langle\delta_3,\ \delta_1\delta_2\rangle,\qquad
\sqrt{\mathcal{J}_1}=\langle\delta_3\rangle\cap\langle\delta_1,\delta_2\rangle
\subsetneq \langle\delta_1,\delta_2,\delta_3\rangle
\]
Hence, distinct initial states differing in one of $\delta_1$ or $\delta_2$ 
produce identical outputs, confirming the loss of observability.

\begin{figure}[h]
        \centering
        \includegraphics[height=4.5cm]{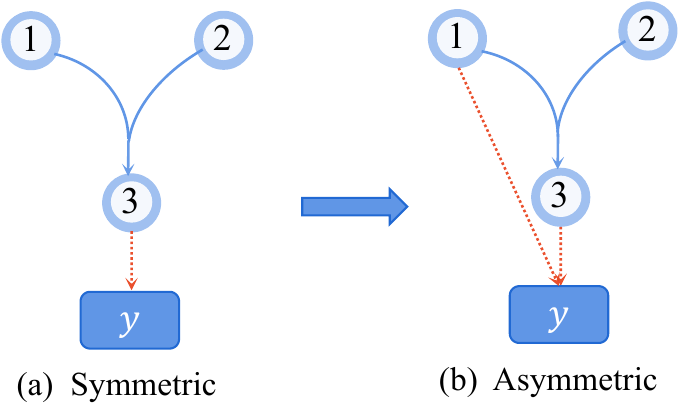}
        \caption{Examples with symmetric and asymmetric structures}
        \label{symmetric}
\end{figure}

Modify the previous system by adding one additional measurement $y_2=x_1$, i.e.,
\begin{equation}
    \dot{x}_3 = a x_1 x_2, \ y_1=x_3, \ y_2=x_1
\end{equation}
The new hypergraph structure is presented in Fig. \ref{symmetric} (b).
Now the backward closure gives $L_0=\{3,1\}$ and $L_1=\{2\}$, 
and the automorphism group becomes trivial, $\mathrm{Aut}=\{\mathrm{id}\}$. 
All state directions are covered and distinguishable, satisfying both Conditions~C1) and~C2); 
hence the system is structurally observable.

Algebraically, the corresponding polynomial differences are
\[
\Psi_0^{(1)}(\delta)=\delta_3,
\Psi_0^{(2)}(\delta)=\delta_1,
\Psi_1^{(1)}(\delta)=a(\sigma_1\delta_2+\sigma_2\delta_1+\delta_1\delta_2)
\]
where $\delta_1$ and $\delta_3$ are directly identified from $\Psi_0^{(2)}$ and $\Psi_0^{(1)}$, 
and $\delta_2$ appears in $\Psi_1^{(1)}$. Thus
\[
\mathrm{in}_\prec(\mathcal{J}_1)\supseteq
\langle \delta_1,\delta_3,\delta_2^{\alpha_2}\rangle
\Rightarrow
\sqrt{\mathcal{J}_1}=\langle\delta_1,\delta_2,\delta_3\rangle
\]
which confirms global structural observability.
}


\section{Local observability criteria of hypergraphs}

In the above section, global
observability criteria of the higher-order network system on hypergraphs are derived, which can determine whether an initial state can be distinguished in the entire statespace by the system trajectory. 
However, for some control systems in practice, 
it is sufficient to distinguish it 
in the subset of the statespace such as its neighborhood. This leads us to investigate the local observability of hypergraphs in this section.


As a foundation, Herman and Krener have derived the rank condition for nonlinear systems \cite {Hermann1977}. Based on it, the local observability criteria of higher-order network systems on hypergraphs are proposed here, which reveal the relation between dynamics of hypergraphs and the rank condition.

Consider the integrated form (\ref{intsys}) of a higher-order network system (\ref{aa2}).  
The local observability of the system (\ref{intsys}) is determined by the dimension of the space spanned by gradients of the Lie derivatives
\begin{align}
    L:=\sum_{i=1}^{n}f_i\frac{\partial}{\partial{x_i}}+\sum_{j= 1}^{N-1}\sum_{l=1}^{w_i}u_{l}^{(j+1)}\frac{\partial}{\partial{u_{l}^{(j)}}} \label{lo1}
\end{align}
of its output function $h_i(x,u)$.

Define the observable Jacobian matrix of the system (\ref{intsys}) as 
\begin{align}
J:=
\left[
\begin{array}{cccc}
    \frac{\partial{L_f^0h_1}}{\partial{x_1} }& \frac{\partial L_f^0h_1}{\partial x_2} & \cdots & \frac{\partial L_f^0h_1}{\partial x_n}\\
     \frac{\partial L_f^0h_2}{\partial x_1} & \frac{\partial L_f^0h_2}{\partial x_2} & \cdots & \frac{\partial L_f^0h_2}{\partial x_n} \\
     \vdots &\vdots &\vdots &\vdots \\
     \frac{\partial L_f^0h_q}{\partial x_1} & \frac{\partial L_f^0h_q}{\partial x_2} & \cdots & \frac{\partial L_f^0h_q}{\partial x_n}\\
      \frac{\partial L_f^1h_1}{\partial x_1} & \frac{\partial L_f^1h_1}{\partial x_2} & \cdots & \frac{\partial L_f^1h_1}{\partial x_n}\\
     \vdots &\vdots &\vdots &\vdots \\
     \frac{\partial L_f^1h_q}{\partial x_1} & \frac{\partial L_f^1h_q}{\partial x_2} & \cdots & \frac{\partial L_f^1h_q}{\partial x_n} \\
    \vdots &\vdots &\vdots &\vdots \\
     \frac{\partial L_f^{n-1}h_q}{\partial x_1} & \frac{\partial L_f^{n-1}h_q}{\partial x_2} & \cdots & \frac{\partial L_f^{n-1}h_q}{\partial x_n}\label{lo2}
\end{array}
\right]_{nq\times n}
\end{align}

\begin{lem}[local observability] \label{lem4}The higher-order network system (\ref{aa2}) is locally observable if and only if the Jacobian matrix $J$ satisfies $\mathrm{rank}(J)=n$. \end{lem}

\begin{rem} \label{rem6} The above lemma is a local observability criterion for the higher-order network system \eqref{aa2} which can be directly obtained from the existing results of observability for nonlinear systems \cite{hermann77, zeitz84, liu13}. This lemma will be utilized to further associate the rank condition with the parameters of higher-order networks in the following subsections. Three scenarios will be considered, respectively: the higher-order network system without input, the system with input, and the system with direct transmission.
\end{rem}

\subsection{Local observability of higher-order networks without input and direct transmission}

First, consider the higher-order network system (\ref{aa2}) without input. According to (\ref{a3}), 
the Kronecker form can be written as
\begin{eqnarray}\label{V1}
\left\{\begin{aligned} 
\dot{{x}}=&\sum_{k=2}^{c}\mathcal{A}_k{x}^{[k-1]}, \ \ {x}(0)={x}_o \\
y_i=&\sum_{k=1}^{c}\mathcal{C}_{i,k}{x}^{[k]}, \ \ i=1,\cdots,q \end{aligned}\right.
\end{eqnarray}
where $\mathcal{A}$ and $\mathcal{C}$ are defined in (\ref{a3}). 

According to the property of products of matrices that $(A\otimes B)(C\otimes D)=(AB)\otimes (CD)$, it holds that
\begin{align}
{L_f^0h_i}&=\sum_{k=1}^{c}\mathcal{C}_{i,k}{x}^{[k]}\nn
\end{align}
\begin{align}
{L_f^1h_i}&=\frac{d}{d t}\sum_{k=1}^{c}\mathcal{C}_{i,k}{x}^{[k]}=\sum_{k=1}^{c}\mathcal{C}_{i,k}\frac{d}{d t}\left(\overbrace{x \otimes x \otimes \cdots\otimes x}^{k \ \mathrm{times}} \right)\nn\\
&=\sum_{k=1}^{c}\mathcal{C}_{i,k}\bigg(\dot x \otimes \cdots\otimes x + \cdots+ x \otimes \cdots\otimes \dot x  \bigg)\nn\\
&=\sum_{k=1}^{c}\mathcal{C}_{i,k}\bigg(\sum_{m=2}^{c}\mathcal{A}_m{x}^{[m-1]} \otimes \cdots\otimes x + \cdots\nn\\
&\hspace{5mm}+ x \otimes \cdots\otimes \sum_{m=2}^{c}\mathcal{A}_m{x}^{[m-1]} \bigg)\nn\\
&=\sum_{k=1}^{c}\mathcal{C}_{i,k}\sum_{j=1}^{k}\left(x\otimes \cdots\otimes \underbrace{\sum_{m=2}^{c}\mathcal{A}_m{x}^{[m-1]}}_{{j{\rm -th~position}}} \otimes \cdots\otimes x\right)\nn\\
&=\sum_{k=1}^{c}\mathcal{C}_{i,k}\sum_{j=1}^{k}\sum_{m=2}^{c}\left(I\otimes \cdots\otimes \underbrace{\mathcal{A}_m}_{{j{\rm -th~position}}} \otimes \cdots\otimes I\right)\nn\\
&\hspace{5mm}\times{x}^{[k+m-2]}\nn\\
&\vdots\nn\\
{L_f^nh_i}&=\sum_{k=1}^{c}\mathcal{C}_{i,k} \bar{A}_1 \bar{A}_2\cdots \bar{A}_n x^{[k+n(m-2)]}\nn
\end{align}
where 
\begin{align}
\bar{A}_n=\hspace{-0.6cm}\sum_{j=1}^{k+(n-1)(m-2)}\sum_{m=2}^{c}\left(\overbrace{I\otimes \cdots\otimes \hspace{-0.4cm}\underbrace{\mathcal{A}_m}_{j{\rm -th~position}}\hspace{-0.4cm} \otimes \cdots\otimes I}^{k+(n-1)(m-2) \ \mathrm{times}}\right) \label{barA}
\end{align}

Then, the observability matrix $\mathbf{O}(x)$ is defined as
\begin{eqnarray}
\mathbf{O}(x):=\nabla_x\left( \begin{array}{cccc}
\sum_{k=1}^{c}\mathcal{C}_{i,k}{x}^{[k]} \\
\sum_{k=1}^{c}\mathcal{C}_{i,k}\bar{A}_1{x}^{[k+m-2]} \\
\vdots \\
\sum_{k=1}^{c}\mathcal{C}_{i,k} \bar{A}_1 \bar{A}_2\cdots \bar{A}_n x^{[k+n(m-2)]} 
\end{array}\right) \label{maO}
\end{eqnarray}


\begin{thm} \label{thm3} The higher-order network system
\begin{align}\label{V11}\left\{\begin{aligned}
\dot{{x}}&=\sum_{k=1}^{c}\mathbf{A}_k{x}^{k-1} \\
y_i&=\sum_{k=1}^{c}\mathbf{C}_{i,k}{x}^{k}\end{aligned}\right.
\end{align}
is locally observable if and only if $\text{rank}(\mathbf{O}(x))=n$.
\end{thm}

\begin{proof}According to equation \eqref{a3}, system \eqref{V11} can be transformed into \eqref{V1}. Then, based on the analysis in Section V-A, by substituting \eqref{maO} into \eqref{lo2}, Theorem \ref{thm3} can be obtained according to Lemma \ref{lem4}.
\end{proof}

\begin{rem} \label{rem7} The theorem provides a general formulation of the observable Jacobian matrix for a higher-order network system without input. By analyzing the span of the Lie derivative, we identify the underlying rule governing how the multi-order derivatives of the system output evolve with the structure of higher-order networks. This establishes a connection between the general rank condition of nonlinear systems and the structural properties of higher-order networks. Furthermore, when the system order reduces to $1$, corresponding to the general linear system $\dot{x} = Ax, y = Cx$, the observability matrix $\mathbf{O}$ in \eqref{maO} simplifies to $(C, CA, CA^2, \dots, CA^n)$, which aligns with the classical observability criterion for linear systems. \end{rem}



\subsection{Local observability of higher-order networks with input}

Consider a higher-order network system with input.
According to (\ref{a3}), the Kronecker form can be written as 
\begin{eqnarray}\label{V2}
 \left\{\begin{aligned} 
\dot{{x}}=&\sum_{k=2}^{c}\mathcal{A}_k{x}^{[k-1]}+\sum_{j=1}^{m}\sum_{k=2}^{c}\mathcal{B}_{k,j}{x}^{[k-1]}u_j\\
y_i=&\sum_{k=1}^{c}\mathcal{C}_{i,k}{x}^{[k]}, \ \ i=1,\cdots,q \end{aligned}\right.
\end{eqnarray}
where $\mathcal{A}$, $\mathcal{B}$, and $\mathcal{C}$ are defined in (\ref{a3}). 
It holds that 
\begin{align}
{L_f^0h_i}&=\sum_{k=1}^{c}\mathcal{C}_{i,k}{x}^{[k]}\nn\\
{L_f^1h_i}&=\frac{d}{d t}\sum_{k=1}^{c}\mathcal{C}_{i,k}{x}^{[k]}=\sum_{k=1}^{c}\mathcal{C}_{i,k}\frac{d}{d t}\left(\overbrace{x \otimes x \otimes \cdots\otimes x}^{k \ \mathrm{times}} \right)\nn\\
&=\sum_{k=1}^{c}\mathcal{C}_{i,k}\bigg(\dot x \otimes \cdots\otimes x + \cdots+ x \otimes \cdots\otimes \dot x  \bigg)\nn\\
&=\sum_{k=1}^{c}\mathcal{C}_{i,k}\bigg(\sum_{\ell=2}^{c}\mathcal{A}_{\ell}{x}^{[\ell-1]} \otimes \cdots\otimes x + \cdots\nn\\
&\hspace{5mm}+ x \otimes \cdots\otimes \sum_{\ell=2}^{c}\mathcal{A}_{\ell}{x}^{[\ell-1]} \nn\\
&\hspace{5mm}+\sum_{j=1}^{m}\sum_{\ell=2}^{c}\mathcal{B}_{\ell,j}{x}^{[\ell-1]}u_j \otimes \cdots\otimes x + \cdots\nn\\
&\hspace{5mm}+ x \otimes \cdots\otimes \sum_{p=1}^{m}\sum_{\ell=2}^{c}\mathcal{B}_{\ell,p}{x}^{[\ell-1]}u_p \otimes \cdots\otimes x\bigg)\nn\\
&=\sum_{k=1}^{c}\sum_{\ell=2}^{c}\mathcal{C}_{i,k}\tilde A_1 {x}^{[k+\ell-2]}\nn
\end{align}
where 
\begin{align}
\tilde A_1= \sum_{j=1}^{k}\left(I\otimes \cdots\otimes \underbrace{\mathcal{A}_{\ell}+\sum_{p=1}^{m}\mathcal{B}_{\ell,p}u_p}_{j{\rm -th~position}} \otimes \cdots\otimes I\right) \nn
\end{align}

First, recall the general Leibniz Rule \cite{Leibniz}. The $n$-th derivative of the product of two functions $F$ and $G$ can be obtained as $(FG)^{(n)}=\sum_{\xi=0}^{n}C_{n}^{\xi}F^{(n-\xi)}G^{(\xi)}$ where $C_{n}^{\xi}=\frac{n!}{\xi!{(n-\xi)}!}$ is the binomial coefficient and $F^{(s)}$ denotes the $s$-th derivative of $F$.
Therefore, ${L_f^nh_i}$ can be derived as
\begin{align}
{L_f^nh_i}=\sum_{k=1}^{c}\sum_{\ell=2}^{c}\mathcal{C}_{i,k}\sum_{\xi=0}^{n-1}C_{n-1}^{\xi}\tilde A_1^{(n-1-\xi)} ({x}^{[k+\ell-2]})^{(\xi)} \label{lnf}
\end{align}

In \eqref{lnf}, define $A_{Bu}=\mathcal{A}_{\ell}+\sum_{p=1}^{m}\sum_{\ell=2}^{c}\mathcal{B}_{\ell,p}u_p$. Then
the $s$-th derivative of $A_{Bu}$ is
\begin{align}
A_{Bu}^{(s)}=\mathcal{A}_{\ell}+\sum_{p=1}^{m}\frac{d^s\mathcal{B}_{\ell,p}}{d t}u_p
\end{align}

According to the definition of $\tilde A_1$, one can deduce that the $s$-th derivative of $A_1$ to $A_n$ are 
\begin{align}
\tilde A_1^{(s)}&= \sum_{j=1}^{k}\left(I\otimes \cdots\otimes 
\underbrace{A_{Bu}^{(s)}}_{j{\rm -th~position}} \otimes \cdots\otimes I\right) \nn\\
&\vdots\nn\\
\tilde{A}^{(s)}_n&=\hspace{-0.6cm}\sum_{j=1}^{k+(n-1)(\ell-2)}\sum_{\ell=2}^{c}\left(\overbrace{I\otimes \cdots\otimes \hspace{-0.4cm}\underbrace{A_{Bu}^{(s)}}_{j{\rm -th~position}}\hspace{-0.4cm} \otimes \cdots\otimes I}^{k+(n-1)(\ell-2) \ \mathrm{times}}\right)  \nn
\end{align}

Furthermore, it holds that
\begin{align}
{x}^{[k+\ell-2]'}=\tilde A_2 x^{[k+2(\ell-2)]} \nn
\end{align}

This leads to the following equations
\begin{align}
({x}^{[k+\ell-2]})^{(n-1)}
&=(\tilde A_2 x^{[k+2(\ell-2)]})^{(n-2)}\nn\\
&=\sum_{\xi=0}^{n-2}C_{n-2}^{\xi}\tilde A_2^{(n-2-\xi)} (x^{[k+2(\ell-2)]})^{(\xi)} 
\nn\\
(x^{[k+2(\ell-2)]})^{(n-2)}&= \sum_{\xi=0}^{n-3}C_{n-3}^{\xi}\tilde A_3^{(n-3-\xi)} (x^{[k+3(\ell-2)]})^{(\xi)}\nn \\
(x^{[k+3(\ell-2)]})^{(n-3)}&=
\sum_{\xi=0}^{n-4}C_{n-4}^{\xi}\tilde A_4^{(n-4-\xi)} (x^{[k+4(\ell-2)]})^{(\xi)} \nn\\
&\vdots\nn\\
(x^{[k+(n-1)(\ell-2)]})'&=\tilde A_{n}x^{[k+n(l-2)]}\nn
\end{align}

Therefore, ${L_f^nh_i}$ can be derived as
\begin{align}
{L_f^nh_i}=&\sum_{k=1}^{c}\sum_{\ell=2}^{c}\sum_{\gamma=1}^{n}\mathcal{C}_{i,k}\sum_{\xi_1=0}^{n-1}C_{n-1}^{\xi_1}\tilde A_1^{(n-1-\xi_1)}\nn\\
&\sum_{\xi_2=0}^{\xi_1-1}C_{\xi_1-1}^{\xi_2}\tilde A_2^{(\xi_1-1-\xi_2)}\cdots\times\nn\\
& \sum_{\xi_{\gamma-1}=0}^{\xi_{\gamma-2}-1}C_{\xi_{\gamma-2}-1}^{\xi_{\gamma-1}}\tilde A_{\gamma-1}^{(\xi_{\gamma-2}-1-\xi_{\gamma-1})} \tilde A_{\gamma}x^{[k+\gamma(l-2)]}\nn
\end{align}
where $\xi_t\in N, 0\leq \xi_t\leq t-2,t\in\{2,3,\cdots,n\}$, and $\xi_r=0$.

\begin{figure*}[h]
\begin{align}\label{O1}
\mathbf{O}_1(x):=&\nabla_x\left( \begin{matrix}
\ts\sum_{k=1}^{c}\mathcal{C}_{i,k}{x}^{[k]} \\
\ts\sum_{k=1}^{c}\sum_{\ell=2}^{c}\mathcal{C}_{i,k}\tilde A_1 {x}^{[k+\ell-2]} \\
\ts\vdots \\
\ts\sum_{k=1}^{c}\sum_{\ell=2}^{c}\sum_{\gamma=1}^{n}\mathcal{C}_{i,k}\sum_{\xi_1=0}^{n-1}C_{n-1}^{\xi_1}\tilde A_1^{(n-1-\xi_1)}\\
\ts\sum_{\xi_2=0}^{\xi_1-1}C_{\xi_1-1}^{\xi_2}\tilde A_2^{(\xi_1-1-\xi_2)}\cdots\sum_{\xi_{\gamma-1}=0}^{\xi_{\gamma-2}-1}C_{\xi_{\gamma-2}-1}^{\xi_{\gamma-1}}\tilde A_{\gamma-1}^{(\xi_{\gamma-2}-1-\xi_{\gamma-1})} \tilde A_{\gamma}x^{[k+\gamma(l-2)]}
\end{matrix}\right)
\end{align}
\end{figure*}

Then, by defining the observability matrix of system \eqref{V2} as $\mathbf{O}_1$ in \eqref{O1}, one can obtain the following result.

\begin{thm}\label{thm4} The higher-order network system
\begin{align}\label{V22}
\left\{\begin{aligned} 
\dot{{x}}=&\sum_{k=2}^{c}\mathbf{A}_k{x}^{k-1} +\sum_{j=1}^{m}\sum_{k=2}^{c}\mathbf{B}_{k,j}{x}^{k-1}u_j  \\ y_i=&\sum_{k=1}^{c}\mathbf{C}_{i,k}{x}^{k} 
\end{aligned}\right.
\end{align}
is locally observable if and only if $\text{rank}(\mathbf{O}_1)=n$.
\end{thm}

\begin{proof}According to equation \eqref{a3}, system \eqref{V22} can be transformed into \eqref{V2}. Based on the analysis in Section V-B, by substituting the definition $\mathbf{O}_1$ in \eqref{O1} into \eqref{lo2}, Theorem \ref{thm4} can be obtained according to Lemma \ref{lem4}.
\end{proof}

\subsection{Local observability of higher-order networks with direct transmission}

Now we consider the local observability of a higher-order network system (\ref{aa2}) with direct transmission. 
According to (\ref{a3}), the Kronecker form can be written as 
\begin{align}\label{V3}
\left\{\begin{aligned}
\dot{{x}}=&\sum_{k=2}^{c}\mathcal{A}_k{x}^{[k-1]} \\
y_i=&\sum_{k=1}^{c}\mathcal{C}_{i,k}{x}^{[k]}+\sum_{l=1}^{w_i}\sum_{k=1}^{c}\mathcal{D}_{i,k,l}{x}^{[k]}u_l 
\end{aligned}\right.
\end{align}

Based on (\ref{lo1}) and the Kronecker form (\ref{a3}), it yields that
\begin{align}
{L_f^0h_i}&=\sum_{k=1}^{c}\mathcal{C}_{i,k}{x}^{[k]}+\sum_{l=1}^{w_i}\sum_{k=1}^{c}\mathcal{D}_{i,k,l}{x}^{[k]}u_l\nn\\
{L_f^1h_i}&=\sum_{k=1}^{c}\mathcal{C}_{i,k}\frac{d {x}^{[k]}}{d t}+\sum_{l=1}^{w_i}\sum_{k=1}^{c}\mathcal{D}_{i,k,l}\frac{d {x}^{[k]}}{d t}u_l\nn\\
&\hspace{5mm}+\sum_{l=1}^{w_i}\sum_{k=1}^{c}\mathcal{D}_{i,k,l}{x}^{[k]}\frac{d u_l}{d t} \nn\\
&=\left(\sum_{k=1}^{c}\mathcal{C}_{i,k}+\sum_{l=1}^{w_i}\sum_{k=1}^{c}\mathcal{D}_{i,k,l}u_l\right)\sum_{m=2}^{c}\sum_{j=1}^{k}\nn\\
&\hspace{5mm}\left(I\otimes \cdots\otimes \underbrace{\mathcal{A}_m}_{{j{\rm -th~position}}} \otimes \cdots\otimes I\right){x}^{[k+m-2]}\nn\\
&\ \ \ \ +\sum_{l=1}^{w_i}\sum_{k=1}^{c}\mathcal{D}_{i,k,l}{x}^{[k]}\frac{d u_l}{d t}\nn\\
&\hspace{3mm}\vdots \nn
\end{align}
\begin{align}
{L_f^nh_i}&=\!a_1\!\!\left(\sum_{k=1}^{c}\mathcal{C}_{i,k}\!+\!\!\sum_{l=1}^{w_i}\sum_{k=1}^{c}\mathcal{D}_{i,k,l}u_l\!\right)\!\!\bar A_1\bar A_2\cdots\bar A_n x^{[k+n(m-2)]}\nn\\
&\ \ \ \ +a_2\sum_{l=1}^{w_i}\sum_{k=1}^{c}\mathcal{D}_{i,k,l}\bar A_1\cdots \bar A_{n-1}x^{[k+(n-1)(m-2)]}\frac{d u}{d t}\nn\\
&\ \ \ \ +a_3\sum_{l=1}^{w_i}\sum_{k=1}^{c}\mathcal{D}_{i,k,l}\bar A_1\cdots \bar A_{n-2}x^{[k+(n-2)(m-2)]}\frac{d^2 u}{d t}\nn\\
&\ \ \ \ +\cdots +a_n\sum_{l=1}^{w_i}\sum_{k=1}^{c}\mathcal{D}_{i,k,l}\bar A_1x^{[k]}\frac{d^n u}{d t}\nn
\end{align}
where $\bar A_n$ is defined in (\ref{barA}). $a_1, a_2,\cdots, a_n$ correspond to the terms in the $n$-th raw of the Yanghui triangle in turn. According to the property of Yanghui triangle, they are equal to the coefficients in the expansion of $(a+b)^{n-1}$. 


Then, by defining the observability matrix $\mathbf{O}_2(x)$ of system \eqref{V3}  as \eqref{O2}, one can obtain the following result.

\begin{thm}\label{thm5}  The system on higher-order network 
\begin{align}\label{V33}
\left\{\begin{aligned} 
\dot{{x}}&=\sum_{k=1}^{c}\mathbf{A}_k{x}^{k-1} \\ y_i&=\sum_{k=1}^{c}\mathbf{C}_{i,k}{x}^{k}+\sum_{l=1}^{w_i}\sum_{k=1}^{c}\mathbf{D}_{i,k,l}{x}^{k}u_l 
\end{aligned}\right.
\end{align}
is locally observable if and only if $\text{rank}(\mathbf{O}_2)=n$.
\end{thm}

\begin{proof}According to equation \eqref{a3}, system \eqref{V33} can be transformed into \eqref{V3}. Based on the analysis in Section VII-C, by substituting the definition  $\mathbf{O}_2$ in \eqref{O2} into \eqref{lo2}, Theorem \ref{thm5}  can be obtained according to Lemma \ref{lem4}.
\end{proof}

\begin{figure*}[h]
\begin{align}
\label{O2}
\mathbf{O}_2(x):=&\nabla_x\left( \begin{matrix}
\sum_{k=1}^{c}\mathcal{C}_{i,k}{x}^{[k]}+\sum_{l=1}^{w_i}\sum_{k=1}^{c}\mathcal{D}_{i,k,l}{x}^{[k]}u_l \\
\left(\sum_{k=1}^{c}\mathcal{C}_{i,k}+\sum_{l=1}^{w_i}\sum_{k=1}^{c}\mathcal{D}_{i,k,l}u_l\right)\sum_{m=2}^{c}\sum_{j=1}^{k}\bar A_1 {x}^{[k+m-2]} +\sum_{l=1}^{w_i}\sum_{k=1}^{c}\mathcal{D}_{i,k,l}{x}^{[k]}\frac{d u_l}{d t} \\
\vdots \\
\left(\sum_{k=1}^{c}\mathcal{C}_{i,k}+\sum_{l=1}^{w_i}\sum_{k=1}^{c}\mathcal{D}_{i,k,l}u_l\right)\bar A_1\bar A_2\cdots\bar A_n x^{[k+n(m-2)]}\\
+\sum_{p=1}^{n-1}a_{p+1}
\sum_{l=1}^{w_i}\sum_{k=1}^{c}\mathcal{D}_{i,k,l}\bar A_1\cdots \bar A_{n-p}x^{[k+(n-p)(m-2)]}\frac{d^p u}{d t}
\end{matrix}\right)
\end{align}
\end{figure*}

\begin{rem} \label{rem8} The study of local observability for higher-order network systems on hypergraphs above is crucial for understanding the system's ability to distinguish between different initial states within a localized region of the state space. Unlike global observability, which requires observability across the entire state space, local observability focuses on specific states and smaller neighborhoods. This is relevant in practical scenarios where it may not be feasible or necessary to assess the entire system globally. 
Rank criteria in Theorems 3-5 can be used to test whether the nodes/hyperedges can be measured or not in a system on hypergraph.
Based on the proposed specific formulas of the observability matrices for different higher-order network systems, by using  searching approach such as the greedy heuristic approach in \cite{Joshua23}, the minimum set of observable nodes can be estimated. Furthermore, the measurement information will be applied to help design the higher-order feedback control of systems on hypergraph. 
\end{rem}

\section{Examples}

\subsection{Numerical case}

In the following, an numerical example is given to demonstrate the testability of the aforementioned criteria for higher-order network systems.

\textit{Example.} Consider the system on a higher-order network 
\begin{align}\label{E1}
\left\{\begin{aligned} 
\dot x =\mathbf{A}_3x^2+\mathbf{A}_2x\\
y=\mathbf{C}_3x^3+\mathbf{C}_2x^2 
\end{aligned}\right.
\end{align}
where ${(\mathbf{A}_3)}_{123}={(\mathbf{A}_3)}_{132}={(\mathbf{A}_3)}_{213}={(\mathbf{A}_3)}_{231}={(\mathbf{A}_3)}_{312}={(\mathbf{A}_3)}_{321}=\frac{1}{2}$, ${(\mathbf{A}_2)}_{12}={(\mathbf{A}_2)}_{21}=1$. ${(\mathbf{C}_3)}_{123}={(\mathbf{C}_3)}_{132}={(\mathbf{C}_3)}_{213}={(\mathbf{C}_3)}_{231}={(\mathbf{C}_3)}_{312}={(\mathbf{C}_3)}_{321}=\frac{1}{2}$ and ${(\mathbf{A}_2)}_{12}={(\mathbf{A}_2)}_{21}=1$.

According to (\ref{product}), the system (\ref{E1}) can be further transformed as follows
\begin{align}
\dot x &=
\left(\begin{array}{ccc}
x_2x_3+x_2\\
x_1x_3+x_1\\
x_1x_2 \label{ex1}
\end{array}\right)\\
y&=x_1x_2x_3+x_1x_2 \label{ex2}
\end{align}
where $x:=(x_1,x_2,x_3)^T$.

For the initial state pair $(\xi,\eta)$, one can obtain that 
\begin{align}
    J_0&=\langle \xi_1\xi_2(\xi_3+1)-\eta_1\eta_2(\eta_3+1)\rangle \nn\\
    J_1&=J_0+\langle (\xi_2\xi_3+\xi_2)^2+(\xi_1\xi_3+\xi_1)^2+\xi_1^2\xi_2^2\nn\\
    &\hspace{5mm}-\left((\eta_2\eta_3+\eta_2)^2+(\eta_1\eta_3+\eta_1)^2+\eta_1^2\eta_2^2\right)\rangle \nn\\
    J_2&=J_1+\langle 2\xi_1\xi_2\left((\xi_3+1)^2+\xi_2^2\right)(\xi_3+1)\nn\\
    &\hspace{5mm}+2\xi_1\xi_2\left((\xi_3+1)^2+\xi_1^2\right)(\xi_3+1)\nn\\
&\hspace{5mm}+\big(2(\xi_2\xi_3+\xi_2)\xi_2+2(\xi_1\xi_3+\xi_1)\xi_1\big)\xi_1\xi_2\nn\\
&\hspace{5mm}+2(\xi_3+1)\xi_1\xi_2(\xi_1^2+\xi_2^2)-\big(2\eta_1\eta_2\left((\eta_3+1)^2+\eta_2^2\right)\nn\\
&\hspace{5mm}\times(\eta_3+1)+2\eta_1\eta_2\left((\eta_3+1)^2+\eta_1^2\right)(\eta_3+1)\nn\\ &\hspace{5mm}+\left(2(\eta_2\eta_3+\eta_2)\eta_2+2(\eta_1\eta_3+\eta_1)\eta_1\right)\eta_1\eta_2\nn\\ &\hspace{5mm}+2(\eta_3+1)\eta_1\eta_2(\eta_1^2+\eta_2^2) \big)\rangle \nn\\
    &=J_1+\langle \xi_1\xi_2(\xi_3+1)\left(4(\xi_3+1)^2+6\xi_1^2+6\xi_2^2\right)\nn\\ &\hspace{5mm}-\left(\eta_1\eta_2(\eta_3+1)\left(4(\eta_3+1)^2+6\eta_1^2+6\eta_2^2\right) \right)\rangle \nn
\end{align}
Therefore, it holds that 
\begin{align}
    J_0\subset{J_1}=J_2=J
\end{align}
For any fixed initial state $\eta=\sigma$, according to Lemma \ref{lem1}, it yields that 
\begin{align}
    \mathcal{J}=&\langle \xi_1\xi_2(\xi_3+1)-\sigma_1\sigma_2(\sigma_3+1), (\xi_2\xi_3+\xi_2)^2\nn\\
    &\hspace{5mm}+(\xi_1\xi_3+\xi_1)^2+\xi_1^2\xi_2^2
    -(\sigma_2\sigma_3+\sigma_2)^2 \nn\\
    &\hspace{5mm}-(\sigma_1\sigma_3+\sigma_1)^2-\sigma_1^2\sigma_2^2\rangle 
\end{align}
Based on the theory of solutions for equations with multiple varieties, the set $\mathbf{V}(\mathcal{J})$ has multiple different elements. This means that there always exists element $\xi\neq\sigma$ in set $\mathbf{V}(\mathcal{J})$ such that $\mathbf{V}(\mathcal{J})\neq\mathbf{V}(\ell)$. Then according to Lemma \ref{lem1}, any initial state of the system (\ref{ex1})-(\ref{ex2}) is undistinguished. 
\hspace*{\fill}$\Box$

\subsection{Higher-order competitive population model}

Consider a competitive population model with third-order interactions \cite{Bairey, cuiLV}. The form of such a system is described as follows
\begin{align}
\dot x_i=r_ix_i \left(1+\sum_{j=1}^{n}{(\mathbf{E})}_{ji}x_j +\sum_{j=1}^{n}\sum_{k=1}^{n}{(\mathbf{F})}_{jki}x_jx_k\right) \nn
\end{align}
where $x_i$ and $r_i$ are the abundance and per-capita intrinsic growth rate such as the reproduction and mortality rates of species $s_i$, respectively. $\mathbf{E}$ and $\mathbf{F}$ are interaction tensors that capture second-order and third-order interactions. ${(\mathbf{E})}_{ji}\leq 0$ denotes the species $j$’s additive influence on $i$ and ${(\mathbf{F})}_{jki}\leq 0$ denotes the species $j$ and $k$’s joint non-additive influence on $i$. In this case, all species compete with each other. Consider the following specific population model among three species
\begin{align}\label{comp}
\left\{\begin{aligned} 
\dot x_1=&x_1-4x_1x_2 \\
\dot x_2=&x_2-x_2x_3\\
\dot x_3=& x_3-x_2x_3-x_1x_2x_3 \\
y=&x_2 \end{aligned}\right.
\end{align}

The tensor-based form of the above population system \eqref{comp} can be written as the system in Theorem \ref{thm3}
\begin{align}
\left\{\begin{aligned} 
\dot{x}&=\mathbf{A}_2 {x}+\mathbf{A}_3 {x}^2 +\mathbf{A}_4 {x}^3 \nn\\
y&=\mathbf{C}_1 {x}\nn
\end{aligned}\right.
\end{align}
where ${(\mathbf{A}_2)}_{11}={(\mathbf{A}_2)}_{22}={(\mathbf{A}_2)}_{33}=1$. ${(\mathbf{A}_3)}_{121}=-4$. ${(\mathbf{A}_3)}_{232}={(\mathbf{A}_3)}_{233}=-1$. ${(\mathbf{A}_4)}_{1233}=-1$. ${(\mathbf{C}_1)}_{2} =1$.

For the initial state pair $(\xi,\eta)$, one can obtain that the ideals $J_i$ as
\begin{align}
J_0=&\langle y(\xi)-y(\eta)\rangle \nn\\
=&\langle \xi_2-\eta_2\rangle \nn\\
J_1=&J_0+\langle y'(\xi)-y'(\eta)\rangle \nn\\
=&J_0+\langle (\xi_2-\xi_2\xi_3)-(\eta_2-\eta_2\eta_3)\rangle \nn\\
J_2=&J_1+\langle y''(\xi)-y''(\eta)\rangle \nn\\
=&J_1+\langle \xi_2(1-\xi_3)^2-\xi_2(\xi_3-\xi_2\xi_3-\xi_1\xi_2\xi_3)-\nn\\
&\eta_2(1-\eta_3)^2-\eta_2(\eta_3-\eta_2\eta_3-\eta_1\eta_2\eta_3)\rangle \nn\\
\cdots& \nn
\end{align}

The Gröbner basis of $J_2$ and $J_3$ are equal and it holds that
\begin{align}
J&=J_2=J_3\nn\\
&=\langle -\eta_1\eta_2^2\eta_3 + \eta_2^2\eta_3\xi_1, -\eta_2 + \xi_2, -\eta_2\eta_3 + \eta_2\xi_3\rangle \nn
\end{align}

Consider the initial state $\sigma=(\sigma_1,\sigma_2,\cdots,\sigma_n)^T$, replacing the known initial state $\sigma$ with $\eta$. $\mathcal{J}=\langle -\sigma_1\sigma_2^2\sigma_3 + \sigma_2^2\sigma_3\xi_1, -\sigma_2 + \xi_2, -\sigma_2\sigma_3 + \sigma_2\xi_3\rangle $. Further, by Hilbert's Nullstellensatz theories \cite{Cox1}, one can deduce that $\mathbf V(\mathcal{J})=\{\sigma\}$ when $\sigma \ne 0 $.
$\ell=\langle \xi_1-\sigma_1,\xi_2-\sigma_2,\cdots, \xi_n-\sigma_n\rangle $ and $\mathbf V(\ell)=\{\sigma\}$.
Therefore, according to Lemma \ref{lem1}, the competitive population system \eqref{comp} is globally observable at $\sigma\ne0$. The system \eqref{comp} is not globally observable at $\sigma=0$.

According to Theorem \ref{thm3}, we test the local observability of system \eqref{comp} by calculating the rank of the observability matrix $\mathbf{O}(x)$. By substituting \eqref{comp} into \eqref{maO}, $\text{rank}(\mathbf{O}(x))=3$ when $x_2\neq 0$ and $x_3\neq 0$, which means the system \eqref{comp} is locally observable if $x_2\neq 0$ and $x_3\neq 0$.

\section{Conclusion}
In this paper, we consider both the global and local observability of higher-order network systems with both inputs and outputs by exploiting algebraic conditions and rank conditions under hypergraph structure, respectively. Based on properties of polynomial ideals and varieties, the global observability criteria of hypergraphs are established. The criteria are related to the hypergraph structure through symmetrization of system tensors. On the other hand, by transforming the system into Kronecker product form, local observability criteria are obtained based on rank conditions of the observability matrix. Both the proposed global and local observability criteria are applied to a competitive population model with third-order interactions, which shows the testability of the criteria. Future work will focus on designing hyperedge estimation and higher-order feedback control of network systems on hypergraphs.

%

\newpage

\vspace{0 cm}
\begin{IEEEbiography}[{\includegraphics[width=1.4in,height=1.35in,clip,keepaspectratio]{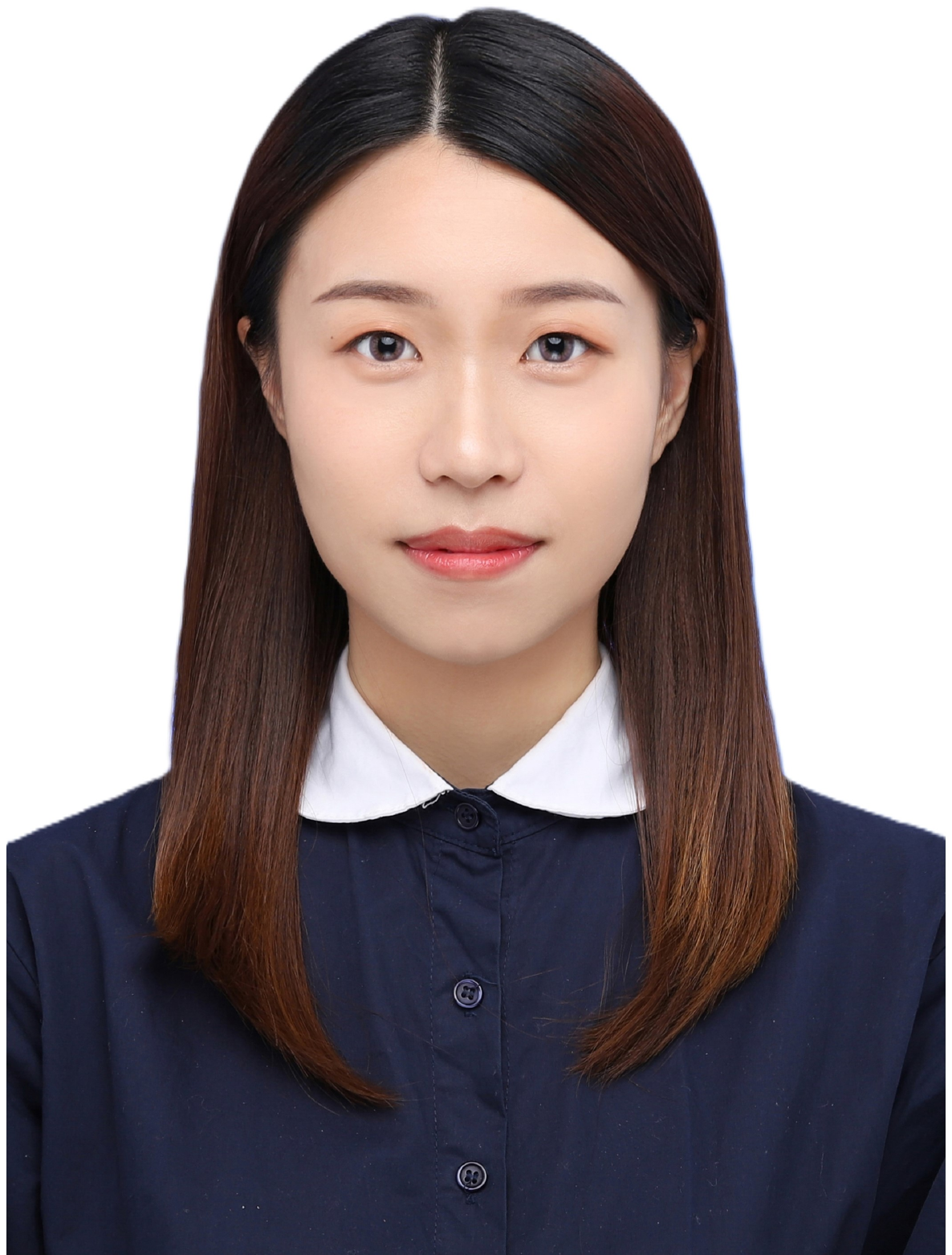}}]{Chencheng Zhang}
received the B.S. degree in automation from the College of Information Engineering, Yangzhou University, Yangzhou, China, in 2017. She was a visiting
 Ph.D. student at the University of Groningen from 2022 to 2023.
She received the Ph.D. degree in control theory and control engineering from Nanjing University of Aeronautics and Astronautics, Nanjing, China, in
2024, where she is currently a Post-Doctoral Researcher.
Her current research focuses on analysis of higher-order network systems, fault estimation and fault-tolerant control for helicopter systems.
\end{IEEEbiography}

\vspace{-1 cm}
\begin{IEEEbiography}[{\includegraphics[width=1in,height=1.25in,clip,keepaspectratio]{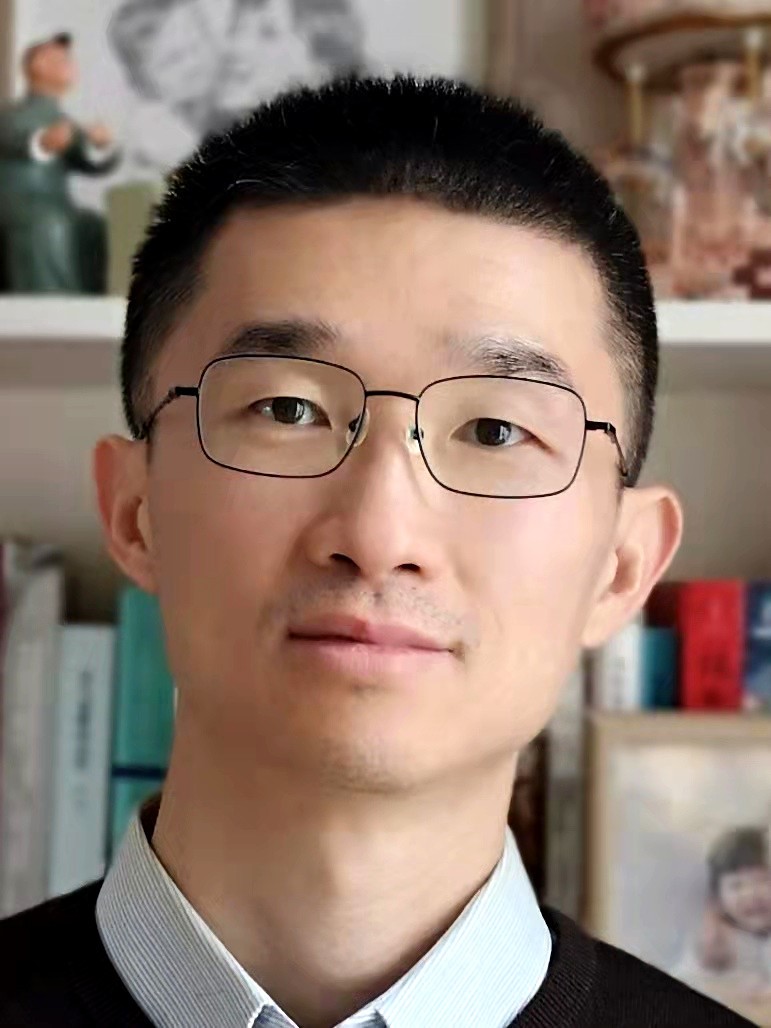}}]{Hao Yang}
(M'13-SM'15) received the B.Sc. degree in electrical automation from Nanjing Tech University, Nanjing, China, in 2004, and the Ph.D. degrees in automatic
control from Université de Lille: Sciences et Technologies, Lille, France, and Nanjing University of Aeronautics and Astronautics (NUAA), Nanjing, China, both in 2009. Since 2010, he has been working with College of Automation Engineering in NUAA, where he has been a Full Professor since 2015. His research interest includes control, optimization, game and fault tolerance of switched and network systems with their aerospace applications.

Dr. Yang was the recipient of the National Science Fund of China for Excellent Young Scholars in 2016, and the Top-Notch Young Talents of Central Organization Department of China in 2017. He has served as Associate Editor for Nonlinear Analysis: Hybrid Systems, Cyber-Physical Systems, Acta Automatica Sinica, and Chinese Journal of Aeronautics. He is a member of the IFAC Technical Committee on Fault Detection, Supervision and Safety of Technical Processes.
\end{IEEEbiography}

\vspace{-1 cm}
\begin{IEEEbiography}[{\includegraphics[width=1in,height=1.25in,clip,keepaspectratio]{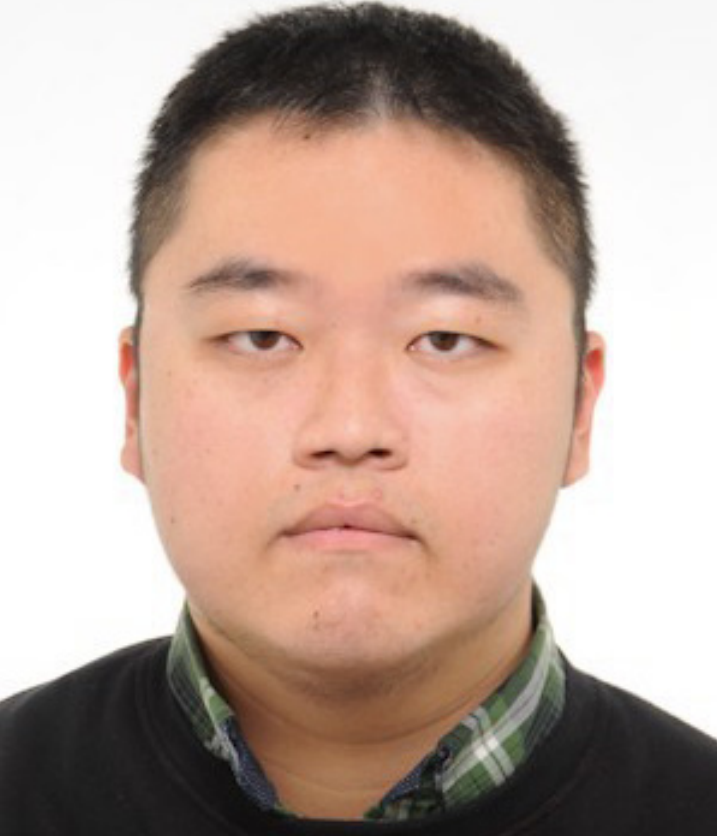}}]{Shaoxuan Cui} is a Ph.D. student with the Bernoulli Institute for Mathematics, Computer Science and Artificial Intelligence of the University of Groningen, the Netherlands. He received his B.Sc. degree in Electrical Engineering and Information Technology from the Technical University of Kaiserslautern, Germany, and Fuzhou University, China in 2018, and his M.Sc. Degree in Electrical Engineering and Information Technology from the Technical University of Munich, Germany, in 2020. His research interests include modeling, analysis, and control of networked systems, especially those networked systems on hypergraphs.
\end{IEEEbiography}

\vspace{-1 cm}
\begin{IEEEbiography}[{\includegraphics[width=1in,height=1.25in,clip,keepaspectratio]{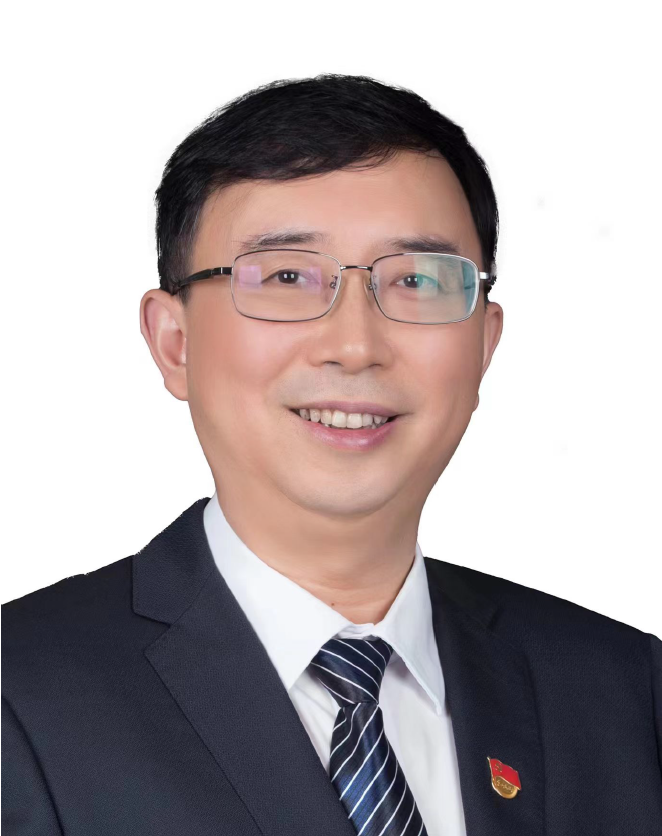}}]{Bin Jiang}
(M'03-SM'05-F'20) received the Ph.D.degree in automatic control from Northeastern University, Shenyang, China, in 1995. He had ever been a Post-Doctoral Fellow, a Research Fellow,
an Invited Professor, and a Visiting Professor in Singapore, France, USA, and Canada, respectively.
He is currently the President of Nanjing University of Aeronautics and Astronautics, Nanjing, China and the Chair Professor of Cheung Kong Scholar Program with the Ministry of Education.
His current research interests include intelligent fault diagnosis and fault tolerant control and their applications to helicopters, satellites, and high-speed trains. He has authored eight books in the related
fields.

Dr. Jiang was a recipient of the Second Class Prize of National Natural Science Award of China. He is a Fellow of Chinese Association of Automation (CAA). He currently serves as an Editor of International Journal of Control, Automation and Systems, and an Associate Editor or an Editorial Board Member for a number of journals, such as IEEE Trans. on Cybernetics, Journal of the Franklin Institute, Neurocomputing, etc. He is a Chair of Control Systems Chapter in IEEE Nanjing Section, and a member of IFAC Technical Committee on Fault Detection, Supervision, and Safety of Technical Processes.
\end{IEEEbiography}

\vspace{-1 cm}
\begin{IEEEbiography}[{\includegraphics[width=1in,height=1.25in,clip,keepaspectratio]{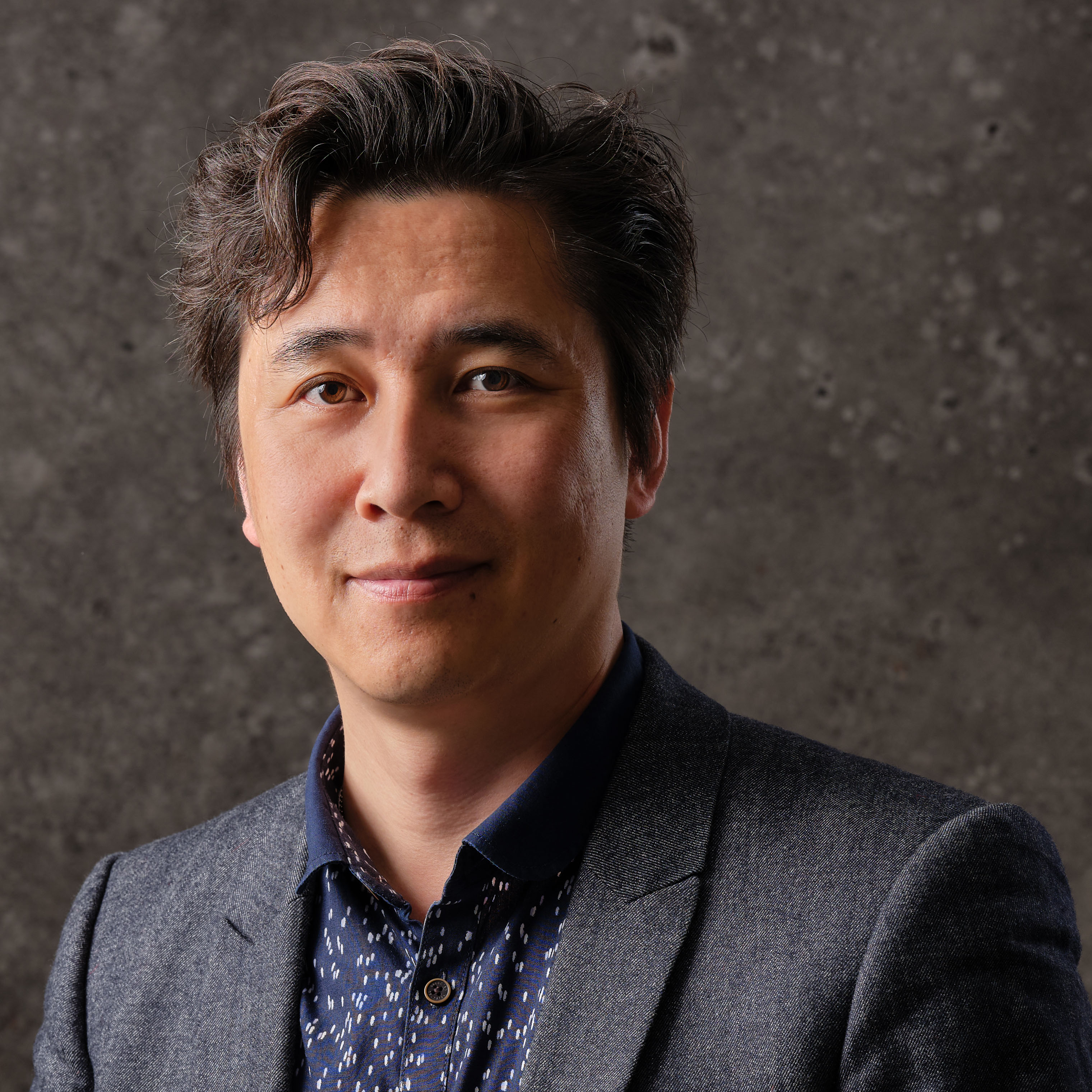}}]{Ming Cao} (F'22)
is a professor of networks and robotics at the University of Groningen, the Netherlands. He received the Bachelor degree in 1999 and the Master degree in 2002 from Tsinghua University, China, and the Ph.D. degree in 2007 from Yale University, USA, all in Electrical Engineering. From 2007 to 2008, he was a Research Associate at Princeton University, USA. He worked as a research intern in 2006 at the IBM T. J. Watson Research Center, USA. He is the 2017 and inaugural recipient of the Manfred Thoma medal from the International Federation of Automatic Control (IFAC) and the 2016 recipient of the European Control Award sponsored by the European Control Association (EUCA). He is an IEEE fellow. He is a Senior Editor for Systems and Control Letters, an Associate Editor for IEEE Transactions on Automatic Control, IEEE Transactions on Control of Network Systems and IEEE Robotics and Automation Magazine, and was an associate editor for IEEE Transactions on Circuits and Systems and IEEE Circuits and Systems Magazine. He is a member of the IFAC Conference Board and a vice chair of the IFAC Technical Committee on Large-Scale Complex Systems. His research interests include autonomous agents and multi-agent systems, complex networks and decision-making processes.
\end{IEEEbiography}

\end{document}